\documentclass[sn-mathphys-num]{sn-jnl}

\usepackage[T1]{fontenc}
\usepackage{amssymb,amsmath,latexsym}
\usepackage{amsthm}
\usepackage{mathrsfs}
\usepackage{diagbox}
\usepackage[shortlabels]{enumitem}
\usepackage{color}
\usepackage{comment}
\usepackage{anyfontsize}
\usepackage{xcolor}
\usepackage{stackengine}
\usepackage{calc}
\usepackage{physics}
\usepackage{accents}
\usepackage{hyperref}
\usepackage{appendix}
\usepackage{colortbl}
\usepackage{booktabs}
\usepackage{tabularx}
\usepackage{tikz-cd}
\usepackage{cleveref}
\usepackage{caption}

\usepackage[most]{tcolorbox}
\usepackage{hyperref}
\hypersetup{
  colorlinks   = true,
  urlcolor     = blue,
  linkcolor    = blue,
  citecolor    = blue
}
\allowdisplaybreaks
\newcommand{\wt}{\widetilde}

\newcommand{\mb}{\mathbb}
\newcommand{\mc}{\mathcal}

 \definecolor{thmblue}{HTML}{1F4E79} \definecolor{thmblueback}{HTML}{EEF6FF} \definecolor{lemmapurple}{HTML}{5B3A8E} \definecolor{lemmapurpleback}{HTML}{F5F0FF} \definecolor{defgreen}{HTML}{2F6B4F} \definecolor{defgreenback}{HTML}{F0FFF4} \definecolor{conjred}{HTML}{9A3412} \definecolor{conjredback}{HTML}{FFF7ED} \definecolor{remarkgray}{HTML}{555555} \definecolor{remarkgrayback}{HTML}{F7F7F7} 
 \theoremstyle{plain} \newtheorem{theorem}{Theorem}[section] \newtheorem{lemma}[theorem]{Lemma}  \newtheorem{proposition}[theorem]{Proposition}  \newtheorem{conjecture}[theorem]{Conjecture}  \theoremstyle{definition} \newtheorem{definition}[theorem]{Definition}   \theoremstyle{remark}   
 \tcbset{ theoremstylebox/.style={ enhanced, breakable, boxrule=0.7pt, arc=1.2mm, outer arc=1.2mm, left=7pt, right=7pt, top=7pt, bottom=7pt, before skip=10pt, after skip=10pt, borderline west={2.5pt}{0pt}{#1}, } } 
 \tcolorboxenvironment{theorem}{ theoremstylebox=thmblue, colback=thmblueback, colframe=thmblue, } \tcolorboxenvironment{proposition}{ theoremstylebox=thmblue, colback=thmblueback, colframe=thmblue, } \tcolorboxenvironment{prop}{ theoremstylebox=thmblue, colback=thmblueback, colframe=thmblue, } \tcolorboxenvironment{corollary}{ theoremstylebox=thmblue, colback=thmblueback, colframe=thmblue, } \tcolorboxenvironment{lemma}{ theoremstylebox=lemmapurple, colback=lemmapurpleback, colframe=lemmapurple, } \tcolorboxenvironment{claim}{ theoremstylebox=lemmapurple, colback=lemmapurpleback, colframe=lemmapurple, } \tcolorboxenvironment{definition}{ theoremstylebox=defgreen, colback=defgreenback, colframe=defgreen, } \tcolorboxenvironment{assumption}{ theoremstylebox=defgreen, colback=defgreenback, colframe=defgreen, } \tcolorboxenvironment{conjecture}{ theoremstylebox=conjred, colback=conjredback, colframe=conjred, } \tcolorboxenvironment{example}{ theoremstylebox=remarkgray, colback=remarkgrayback, colframe=remarkgray, } \tcolorboxenvironment{remark}{ theoremstylebox=remarkgray, colback=remarkgrayback, colframe=remarkgray, } \tcolorboxenvironment{rem}{ theoremstylebox=remarkgray, colback=remarkgrayback, colframe=remarkgray, }

\renewcommand{\footnoterule}{%
  \kern -3pt
  \hrule width \textwidth height 0.4pt
  \kern 2.6pt
}

\newcommand{\firstpageemails}[1]{%
  \begingroup
  \renewcommand{\thefootnote}{}%
  \footnotetext{#1}%
  \endgroup
}

\DeclareMathOperator{\supp}{supp}

\newcommand{\diag}{\operatorname{diag}}

\newcommand{\Haar}{\mathrm{Haar}}

\renewcommand{\thefootnote}{\arabic{footnote}}
\allowdisplaybreaks

\newcommand{\qd}{\end{proof}\vspace{0.5ex}}
\newcommand{\prf}{\begin{proof}[\bf Proof:]}

\usepackage[normalem]{ulem}

\newcommand{\GL}{\mathrm{GL}}
\newcommand{\U}{\mathrm{U}}

\title{
\centering
The Optimal Rate Function in Covariant Quantum State Tomography}
\author[1,2]{\fnm{Arick} \sur{Grootveld}}
\author[2,3]{Alexander Maloney}
\author[1,2]{Jason Pollack} 
\author[4,5]{Peixue Wu}
\affil[1]{Dept. of Electrical Engineering \& Computer Science, Syracuse University, Syracuse, NY, USA}
\affil[2]{Institute for Quantum \& Information Sciences, Syracuse University, Syracuse, NY, USA}
\affil[3]{Dept. of Physics, Syracuse University, Syracuse, NY, USA}
\affil[4]{Dept. of Applied Mathematics, University of Waterloo, Waterloo, ON, Canada}
\affil[5]{Institute for Quantum Computing, University of Waterloo, Waterloo, ON, Canada}
\date{}

\abstract{
The problem of quantum tomography is to estimate an unknown quantum state $\rho$ from a measurement of $n$ copies of $\rho$.  One can ask which tomography protocol, i.e.\ which choice of multi-copy measurement, gives the best possible estimate of $\rho$.  To do so, we characterize tomography protocols by their \emph{rate function}, which 
governs the exponential rate at which a protocol assigns probability to a particular estimate $\sigma$ of the true state $\rho$.
This rate function is a quantum mechanical generalization of the classical relative entropy between the true state and its estimate, and depends on the choice of protocol.  
It is bounded
by the quantum relative entropy, and we show that this bound is sharp: 
for any $\rho$ and $\sigma$ we construct a family of protocols whose rate functions converge to the quantum relative entropy $D(\sigma\|\rho)$.

We consider the family of covariant tomography protocols; 
these are the basis independent state estimation schemes that assume no prior information about $\rho$ and $\sigma$. 
Keyl described a specific tomography protocol based on Schur sampling, and conjectured that among all covariant 
tomography protocols it has the largest possible rate function for all $\sigma$ and $\rho$. 
We prove this conjecture. 
The resulting rate function is an annealed version of quantum relative entropy, due to the cost of learning the eigenbasis in covariant quantum state tomography.
}
\begin{document}
\maketitle
\firstpageemails{%
\texttt{ \{aegrootv,admalone,japollac\}@syr.edu, p33wu@uwaterloo.ca}.
}

\newpage
\tableofcontents
\section{Introduction}

\subsection{Informal Overview}
The basic task of quantum state tomography is the following: given $n$ copies of an unknown quantum state $\rho$, obtain an accurate estimate for $\rho$.  To do so, one must perform a measurement (or measurements) on these $n$ copies.
Of course, in the limit $n\to \infty$ one can determine $\rho$ perfectly by, for example, measuring its matrix elements to any desired accuracy. 
But in any realistic scenario -- such as where $\rho$ is generated by a physical process which we seek to characterize -- $n$ is finite, and we can obtain only 
an approximate description of $\rho$.  
The natural question is whether there is an \emph{optimal} state estimation scheme. Or more precisely: is there a sequence of measurements -- and description of $\rho$ which one obtains from these measurements -- which gives the best possible approximation to $\rho$?  

The goal of this paper is to answer this question in the affirmative, for a certain definition of ``best possible approximation'' of the unknown state $\rho$.  
Specifically, we will characterize possible estimation schemes in terms of the asymptotic behaviour as $n\to\infty$ of the probability that a scheme makes a large error in its determination of $\rho$.  We will then 
show that an algorithm of Keyl \cite{Keyl_2006} provides the optimal approximation scheme, in the sense that it minimizes these errors.  

To understand this, let us begin by stating the problem of quantum tomography a bit more precisely.
Consider an unknown density matrix \(\rho\in\mathcal D_d\) which we are trying to approximate, where \(\mathcal D_d\) is the state space of a \(d\)-dimensional quantum system.  We are given $n$ copies of the state, described by the tensor product $\rho^{\otimes n}$. 
We wish to perform a measurement of $\rho^{\otimes n}$ in order to obtain an approximate description of $\rho$.  Of course, this approximate description will itself be a density matrix on the original $d$-dimensional quantum system; let us call this approximation $\sigma\in \mathcal D_d$.
What this all means is that the measurement will be described by a positive operator-valued measure (POVM) $M_\sigma^{(n)}$, whose output is labeled by the possible approximations
$\sigma \in \mathcal D_d$ of our unknown density matrix.  
The probability of obtaining a particular estimate $\sigma$ of our original density matrix $\rho$ from this measurement is 
\begin{equation}\label{prob}
    P_n(\sigma| \rho) = \Tr\left({M_\sigma^{(n)} \rho^{\otimes n}}\right).
\end{equation}
This is a probability density\footnote{Strictly speaking, $M_\sigma^{(n)}$ is an operator-valued Radon-Nikodym density with respect to a probability measure.
Similarly, we use density notation for $P_n(\sigma|\rho)$ and 
$I(\sigma\|\rho)$ is the corresponding large-devitation rate function. 
This and other mathematical details are relegated to section \ref{sec:general}.}
on the space of possible outputs $\sigma$, which depends on $\rho$ and on the choice of POVMs $\left\{M_\sigma^{(n)}\right\}$.  
Different algorithms for approximating $\rho$, i.e.\ different tomography protocols, simply amount to different choices of the measurements $\left\{M_\sigma^{(n)}\right\}$.

What we have sketched above is the basic mathematical formulation of quantum tomography: a tomography protocol is a POVM on the $n$-fold tensor product of our Hilbert space, whose outcomes are valued in the space $\mathcal D_d$ of density matrices of the original system.  
Each tomography protocol 
leads to a probability distribution \eqref{prob}
on the space of estimates $\sigma$ of the true state $\rho$. 
We can then compare different tomography protocols by studying the behaviour of this probability distribution $P_n(\sigma | \rho)$.  For example, it is natural to require that $P_n(\sigma | \rho)$ is maximized at the point $\sigma=\rho$, and that as $n$ increases the probability distribution should become increasingly concentrated around this point.

Much of the modern literature on quantum state tomography is formulated in terms of sample complexity, which is a measure of how rapidly the probability density $P_n(\sigma| \rho)$ concentrates around the point $\sigma=\rho$.  More precisely, given a tomography protocol one fixes an accuracy parameter \(\varepsilon>0\),
and a failure probability \(\eta>0\), and asks how large $n$ must be to guarantee that
\[
    \Pr\{\|\sigma-\rho\|_1>\varepsilon\} \equiv \int_{\|\sigma-\rho\|_1>\varepsilon} P_n(\sigma | \rho) d\sigma
    \le \eta .
\]
This line of work has led to sharp bounds for mixed-state tomography, and to
efficient algorithms which achieve nearly optimal or optimal sample complexity
\cite{Haah_2017,o2016efficient,pelecanos2025mixed,scharnhorst2025optimal}.
However, while
sample complexity is an essential benchmark, it gives only a coarse description of the error behavior of a tomography protocol. 
For example, two protocols might have the same sample complexity 
while assigning very different probabilities to particular erroneous estimates at moderate sample sizes.  Moreover, it describes only the extent to which the probability distribution for $\sigma$ concentrates around $\rho$ at large $n$, without taking into account the global structure of the probability distribution $P_n(\sigma| \rho)$.

In this paper we will characterize tomography protocols by studying the full probability distribution $P_n(\sigma| \rho)$ in the large $n$ limit. Specifically, we consider the limit
\begin{equation}\label{rate}
     I(\sigma\|\rho)\equiv -\lim_{n\to\infty} \frac{1}{n} \log P_n(\sigma| \rho),
\end{equation}
when it exists. The function $I(\sigma\|\rho)$ is known as the \emph{rate function} for a tomography protocol, and captures the full asymptotic error landscape:
as \(I(\sigma\|\rho)\) increases, the error where we mistake $\rho$ for $\sigma$
becomes exponentially less likely at large $n$.
Thus, unlike the sample complexity, it allows for asymptotic comparisons of tomography protocols which take into account large deviations of the estimator $\sigma$.  Implicitly in equation \eqref{rate}, we have assumed that the probability scales exponentially with $n$, as $P_n(\sigma|\rho)\approx e^{-n I(\sigma\| \rho)}$.  As in classical statistics, this is often (but not always) the case; this is tantamount to the assumption that $P_n(\sigma|\rho)$ satisfies a Large Deviation Principle (LDP).\footnote{
Details on the LDP are described in Section~\ref{sec:LDP}.}

The definition \eqref{rate} is familiar from the classical theory of empirical measures: if one samples a classical probability distribution many times, then the Sanov theorem says that the probability of observing an empirical distribution decays with rate given by the classical relative entropy. 


The quantum case is more subtle. The rate function
\eqref{rate} depends on the choice of tomography protocol: different POVMs may suppress different erroneous estimates at different exponential rates. Nevertheless, as we review below, it is still the case that the rate is upper bounded by the quantum relative entropy \cite[Theorem~3.3]{Keyl_2006}:
\begin{equation}\label{eq:relative-entropy-upper-intro}
I(\sigma\|\rho) \le D(\sigma\|\rho) := {\rm Tr}\,\sigma\left(\log \sigma - \log \rho\right).
\end{equation}
The bound \eqref{eq:relative-entropy-upper-intro} raises a natural question: is the relative entropy actually achievable as a rate function for quantum tomography?  Our first result answers this question in a pointwise sense. We will define an appropriate set $\mathcal E$ of ``good'' rate functions.\footnote{Specifically, $ \mathcal E$ is the set of rate functions for tomography protocols which are consistent with  LDP, and which obey the unique-zero property that $I(\sigma|\rho)=0$ only if $\sigma=\rho$; this means that errors become exponentially unlikely in the large $n$ limit.  See Section \ref{sec:LDP} for more details.}
Then the relative entropy upper bound is sharp at every pair of states:
\begin{theorem}
\label{thm:relative-entropy-envelope-intro}
For all \(\sigma,\rho\in\mathcal D_d\),
\[
    \sup_{I\in \mc E} I(\sigma\|\rho)=D(\sigma\|\rho).
\]
\end{theorem}
This theorem resolves the pointwise relative entropy envelope predicted in \cite{Keyl_2006}. 
It is important, however, to emphasize the pointwise nature of this result:
the protocol used to approach $D(\sigma\|\rho)$ will depend on the pair $(\sigma,\rho)$. 
This means that a protocol that performs well for one pair $(\sigma,\rho)$ may perform poorly for other pairs $(\sigma',\rho')$.


Our goal is to understand whether there is a protocol which is optimal in a sense that is {\it independent} of the choice of $\rho$ and $\sigma$.
In other words, the ``best'' protocol should maximize $I(\sigma\|\rho)$ without assuming any information about the unknown state $\rho$. Such a protocol should be basis-independent.  We will therefore focus our attention on tomography protocols which are \emph{covariant}, in the sense that a change of basis for $\rho$ will lead to the same change of basis for our estimate $\sigma$. More precisely, covariant protocols transform equivariantly with respect to the $\U(d)$ unitary rotations of the original $d$-state quantum system: for any unitary $U$, we have
\begin{equation}
M^{(n)}_{U\sigma U^\dagger } = U^{\otimes n} M^{(n)}_\sigma (U^\dagger)^{\otimes n}\,.
\end{equation}
We now have a clear mathematical goal: optimize the rate function within the space of covariant tomography protocols.
In doing so, it is natural to decompose the POVMs with respect to both the $\U(d)$ rotations of the $d$-state system, and the $S_n$ symmetry that permutes the $n$ samples.  
The resulting representation theory is straightforward, due to Schur-Weyl duality \cite{schur1901klasse,weyl1946classical}.

In studying this question, we will be able to answer a question first posed by Keyl \cite{Keyl_2006}, who initiated the large-deviation approach to quantum state tomography based on the rate function \eqref{rate}.
Keyl proposed a tomography protocol, using Schur--Weyl duality, in which $\rho$ is estimated through a two-step measurement procedure. First, we estimate the spectrum of $\rho$ by weak Schur sampling, i.e.\ by using the measurement that projects $\rho^{\otimes n}$ onto an irreducible representation of the $S_n$ symmetry that permutes the $n$ samples.  By Schur-Weyl duality, this is also an irreducible representation of $\U(d)$.  So in the second step we estimate the eigenbasis through a covariant measurement built from the
highest-weight vector in this $\U(d)$ representation.  This protocol is covariant, in the sense described above.

Keyl computed the rate function of this protocol. 
Writing 
\begin{equation}\label{basis}
    \sigma=U\diag(x)U^\dagger,
\end{equation}
where the spectrum $x$ is ordered from largest to smallest, the rate function is~\cite[Theorem 3.2]{Keyl_2006}:
\begin{equation}\label{eq:Keyl-rate-intro}
    S^{\mathrm{Keyl}}(\sigma\|\rho)
    =-H(x) - \sum_{k=1}^d (x_k - x_{k+1}) \log \Delta_k(U^\dagger \rho U),\quad x_{d+1} \equiv 0,
\end{equation}
where \(H(x)\) is the Shannon entropy of \(x\), and $\Delta_k(\cdot)$ is the determinant of the leading \(k\times k\) principal minor. 
This reduces to the relative entropy
\(D(\sigma\|\rho)\) when \(\sigma\) and
\(\rho\) commute, but can be smaller otherwise.

Our main result is that Keyl's algorithm is optimal among all covariant quantum tomography protocols.  Specifically, we prove 
\begin{theorem}\label{thm:main-introduction}
If a covariant tomography protocol is consistent, i.e., if for any unknown state $\rho$ and any $\varepsilon>0$, 
\begin{equation}\label{cons}
     \lim_{n\to \infty}\Pr\{\sigma \in N_\varepsilon(\rho)\} = 1,\quad N_\varepsilon(\rho):= \{\tau \in \mc D_d: \|\tau - \rho\|_1 < \varepsilon\},
\end{equation}
and if it satisfies a large deviation principle with rate function $I(\cdot\|\cdot)$, then its rate function is pointwise upper bounded by Keyl's rate:
\[
    I(\sigma\|\rho)
    \le
    S^{\mathrm{Keyl}}(\sigma\|\rho),
\]
for all quantum states $\sigma,\rho$.
\end{theorem}
The only non-trivial assumption in this theorem is the consistency requirement, equation \eqref{cons}.  This is simply the requirement that our tomography protocol ``works,'' in the sense that the probability measure for the estimator $\sigma$ concentrates around $\rho$ as $n\to\infty$.  In other words, it is the requirement that the sample complexity is finite. 
Since Keyl's original protocol achieves the rate
\(S^{\mathrm{Keyl}}\), this theorem proves that Keyl's rate is the optimal large-deviation exponent among consistent covariant tomography protocols. 

This also gives a precise large-deviation formulation of the optimality suggested in recent work on debiased versions of Keyl's algorithm \cite{pelecanos2025debiased}, where the authors ask in what sense Keyl's entangled measurement should be regarded as optimal for full state tomography.

Before describing the proof of this theorem, let us comment on the physical intuition behind this result.
First, we note that $S^{\mathrm{Keyl}}(\sigma \| \rho) \le D(\sigma\|\rho)$, and that this inequality is saturated only when $\sigma$ and $\rho$ commute (as we show in Proposition \ref{prop:Keyl-vs-relative-entropy}). 
Intuitively, this means that the problem of quantum state-estimation is more difficult than its classical counterpart.  The reason is that in the quantum case one must determine not just the eigenvalues of $\rho$, but also its eigenbasis.  It is the difficulty of finding this basis which slows down the rate at which any covariant tomography protocol can determine a state, from $D(\sigma\|\rho)$ to $S^{\rm{Keyl}}(\sigma\|\rho)$.  

In fact, this obstruction is already visible in our formula  \eqref{eq:Keyl-rate-intro} for $S^{\rm{Keyl}}(\sigma\|\rho)$.
To see this, consider the basis where $\sigma$ is diagonal, with eigevalues $x_1\ge x_2\ge \dots \ge x_d$ and corresponding eigenstates $\left\{|i\rangle\right\}_{i=1}^d$.  It is easy to see that 
$\sigma = \sum_{k=1}^d \left(x_k - x_{k+1}\right) \Pi_k
$
where $\Pi_k \equiv \sum_{i=1}^k |i\rangle\langle i|$ is the projector on to the first $k$ eigenstates.  
We can then compare the usual relative entropy
\begin{equation}\label{irel}
D(\sigma\|\rho) = -H(x) - \sum_{k=1}^d \left(x_k - x_{k+1}\right) {\rm Tr}\,\left(\Pi_k \log \rho\, \Pi_k\right)
\end{equation}
to Keyl's rate function
\begin{equation}\label{ikeyl}
S^{\rm{Keyl}}(\sigma\|\rho) = -H(x) - \sum_{k=1}^d \left(x_k - x_{k+1}\right) {\rm Tr}\,\log \left(\Pi_k\, \rho\, \Pi_k\right)
\end{equation}
The only difference between the two is whether the projection $\Pi_k$ is carried out inside or outside the logarithm.  
This distinction can be interpreted as the additional cost of learning the eigenbasis of $\rho$.

This phenomenon is familiar from statistical physics, where one studies the thermodynamic properties of statistical averages.  Quantities like \eqref{irel}, which involve the  average of a logarithm, are known as \emph{quenched} averages.
Quantities like \eqref{ikeyl} involve the logarithm of an average, and are known as \emph{annealed} averages.
In the present case, the ``averaging'' under consideration is the projection onto the most likely $k$-dimensional eigenspace of $\sigma$.\footnote{To make the analogy even sharper, we can regard $H=-\log \rho$ as a ``Hamlitonian'' which drives the learning process, and consider the Slater determinant state $|\Omega_k\rangle\equiv |1\rangle\wedge\dots\wedge|k\rangle\in \Lambda^k\mathbb C^d$ in the $k$-fold exterior product of our Hilbert space.  Then  the traces appearing in \eqref{irel} and \eqref{ikeyl} are the quenched  and annealed free energies, $\langle \Omega_k| \Lambda^k \log \rho |\Omega_k\rangle$ and $\log \langle \Omega_k| \Lambda^k \rho |\Omega_k\rangle $, respectively.  Slater determinant states are labeled by $k$-dimensional subspaces (i.e. elements of the Grassmanian $\mathrm{Gr}_k({\mathbb C}^d)$), so we may interpret these terms as the cost associated with determining successive rank $k$ eigenspaces of $\rho$. }
Quenched averages arise when the  quantities being averaged over do not interact with the dynamics of the system being studied.  Annealed averages arise when the quantities being averaged evolve along with the dynamics of the system. In our analogy with statistical mechanics, the ``dynamics'' is that of the learning process according to a topographic protocol.  This perspective suggests why the optimal rate among covariant protocols is $S^{\rm{Keyl}}(\sigma\|\rho)$ rather than $D(\sigma\|\rho)$: a covariant protocol requires one to learn the eigenspaces at the same time that one learns the spectrum, resulting in an annealed average rather than a quenched one.

We now briefly sketch the main ideas behind the proof of Theorem \ref{thm:relative-entropy-envelope-intro} and Theorem \ref{thm:main-introduction}.
\subsection{Outline of Proof}
\textbf{Sketch of proof of Theorem~\ref{thm:relative-entropy-envelope-intro}.}

The converse part is a simple application of quantum Stein lemma \cite{hiai1991proper,ogawa2000strong}. For achievability, fix a pair \((\sigma,\rho)\) and an exponent \(a<D(\sigma\|\rho)\).  We split the samples into two independent parts. A small fraction is used for a standard consistent tomography protocol, ensuring the ``good'' property for the rate function. The remaining samples are divided into blocks and measured with a Stein test tailored to \((\sigma,\rho)\), whose empirical success frequency carries an exponent larger than \(a\), for example \cite{Hayashi_2001}.  The final estimator keeps the tomographic estimate when this empirical frequency is compatible with it, and otherwise outputs a fixed state \(\omega_*\). This produces a non-covariant tomography protocol whose
rate at the chosen pair exceeds \(a\), see Figure~\ref{fig:noncovariant} for illustration.

\begin{figure}[!t]
    \centering
    \includegraphics[width=1.03\linewidth]{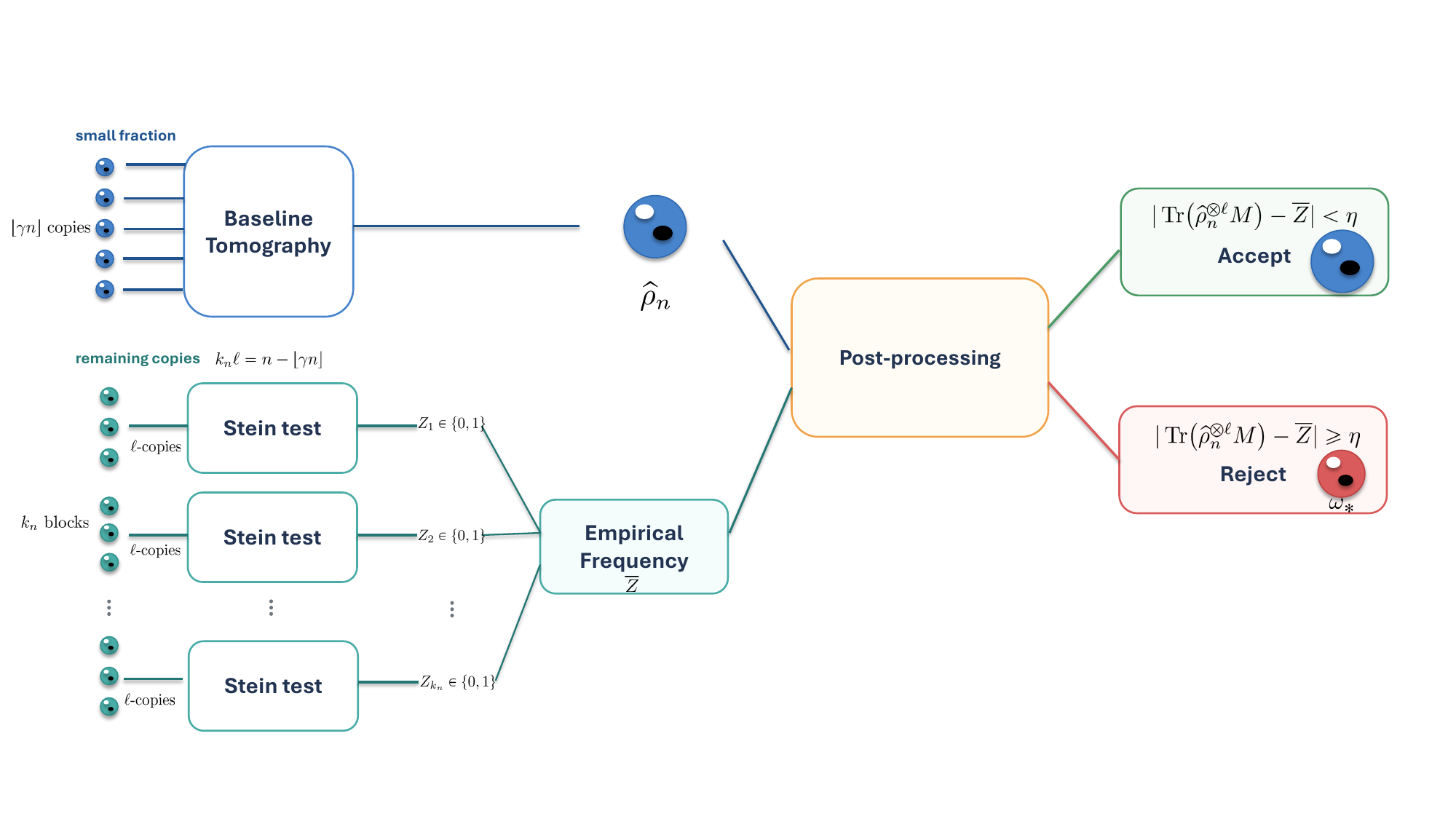}
    \caption{Non-covariant protocol almost achieving the relative entropy as its rate function}
    \label{fig:noncovariant}
\end{figure}

The only technical point is that the final post-processing map has a threshold discontinuity, so the usual contraction principle does not apply directly. We instead use the measurable contraction principle of \cite{Mariani2018}, reviewed in Appendix~\ref{app:LDP}.

\noindent \textbf{Sketch of proof of Theorem~\ref{thm:main-introduction}.}
The proof is divided into three steps:

\noindent \textbf{Step I: A normal form of a general covariant tomography protocol}

Every tomography protocol is i.i.d.-statistically equivalent to one which is invariant under permutations of the $n$ samples. This allows us to characterize the POVM in the Schur basis. A covariant protocol is a POVM indexed by a spectrum $x$ and unitary $U$ for our estimator $\sigma$, as in \eqref{basis}. In the Schur basis, the POVM takes the form given in Figure~\ref{fig:schur-block-diagonal}, where $\lambda^k$ labels different Young diagrams with $n$ boxes and $d$ rows. This characterization is given in Theorem~\ref{thm:characterization-covariant}.
\begin{figure}[!t]
    \centering
    \includegraphics[width=0.99\linewidth]{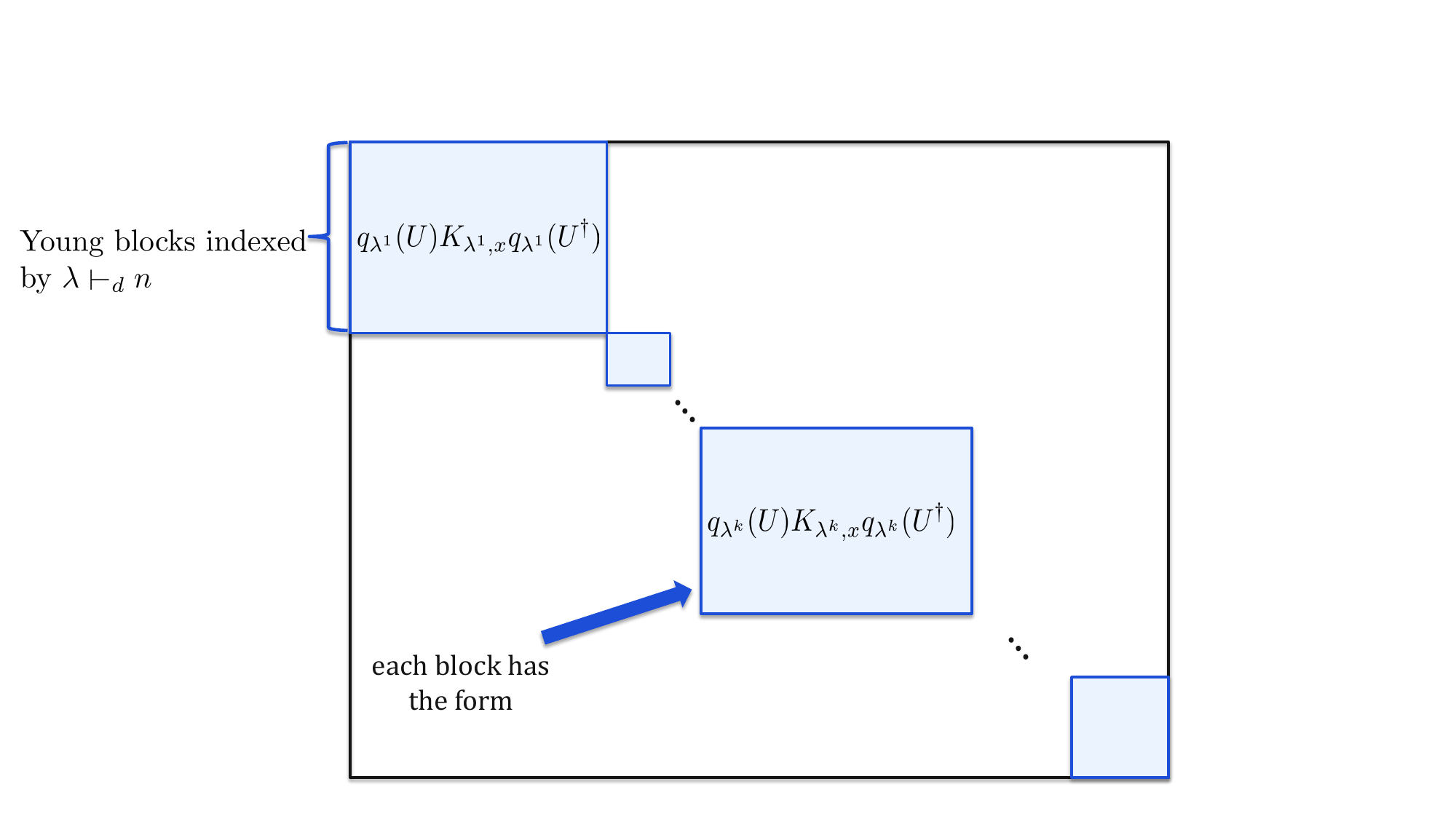}
    \caption{Block diagonal form of a covariant POVM density in the Schur basis.}
    \label{fig:schur-block-diagonal}
\end{figure}









A general tomography protocol is the following:
\begin{tcolorbox}[
    colback=white,
    colframe=black,
    boxrule=0.8pt,
    arc=1pt,
    left=6pt,
    right=6pt,
    top=6pt,
    bottom=6pt
]
Given \(n\)-copies of \(\rho\):
\begin{enumerate}
    \item Apply the Schur transform to write \(\rho^{\otimes n}\) in the Schur basis.
    \item Perform the POVM of Figure~\ref{fig:schur-block-diagonal}, indexed by spectrum \(x\) and unitary \(U\).
    \item Collect the measurement outcomes \(\widehat{x}\) and \(\widehat{U}\), and output
    \[
        \sigma
        =
        \widehat{U}\operatorname{diag}(\widehat{x})\widehat{U}^{\dagger}.
    \]
\end{enumerate}
\end{tcolorbox}
A key point is that, since we are studying covariant protocols, the output gives separate estimates for the spectrum $\widehat{x}$ and the eigenbasis $\widehat{U}$.

\medskip
\noindent
\textbf{Step II: A quantum analogue of the correct empirical type.}

The second step uses consistency. Classically, if the true distribution is \(x\), then a consistent estimator must place almost all of its mass near the empirical type \(x\). In the quantum setting, the Young diagram \(\lambda\) plays the role of the empirical type. Given a state \(\diag(x)^{\otimes n}\), Schur sampling results in a measure concentrated on Young diagrams whose
normalized shape \(\lambda/n\) is close to \(x\). 

Therefore, if the tomography protocol is consistent, then when the true state is \(\diag(x)\) most of the probability 
will come from Young diagrams satisfying \(\lambda/n\approx x\). Since there are only
polynomially many Young diagrams, at least one such Schur block must contribute with the correct exponential order. This is the content of Lemma~\ref{lem:consistent-implication}.

\medskip
\noindent
\textbf{Step III: Weight analysis inside the selected block.}

The final step analyzes the selected Schur block more carefully. Each Schur
module decomposes into weight spaces for the diagonal torus $\mathrm{T}(d)\subseteq \U(d)$. Because the
localized neighborhood of \(\diag(x)\) is invariant under the torus, the
relevant localized operator can be assumed to be block diagonal in these weight
spaces.

Consistency forces a non-negligible contribution from weights whose normalized
weight is close to \(x\). A principal-minor lower bound, Lemma~\ref{lem:priciple-minor-lower-bound}, then compares the value of the same block on an arbitrary state \(\rho\) with its value on \(\diag(x)\). As shown in Lemma~\ref{lem:robust-Keyl-comparison-torus-invariant}, this comparison is exactly where Keyl's expression---involving logarithms of principal minors---appears.

Combining these ingredients gives a lower bound on the probability of 
obtaining an estimate 
$\sigma$ when the true state is \(\rho\). Equivalently, it gives the upper bound
\[
    I(\sigma\|\rho)
    \le
    S^{\mathrm{Keyl}}(\sigma\|\rho),
\]
for the rate function of any consistent covariant protocol satisfying a Large Deviation Principle. Since Keyl's protocol achieves \(S^{\mathrm{Keyl}}\), this proves the optimality of Keyl's rate among covariant tomography protocols.

\subsection{Related work and discussion}



The rate $S^{\mathrm{Keyl}}$ has appeared in several other contexts, often
under a different name.  It is the same quantity as the reverse quantum relative entropy $D_R(\sigma\|\rho):=\lim_{\alpha\to1}D_{RS,\alpha}(\sigma\|\rho)$, where \(D_{RS,\alpha}\) is the reverse sandwiched Rényi divergence introduced in \cite[equation (9)]{audenaert2015alpha}. With this
notation,
\[
    S^{\mathrm{Keyl}}(\sigma\|\rho)=D_R(\sigma\|\rho).
\]
Recent work of Hayashi and Fang \cite{hayashi2026operational} adopts the name reverse quantum relative entropy for this limit and gives it an operational interpretation in composite quantum hypothesis testing. Fix a basis in which the null state \(\rho\) is diagonal, and let \(G\) be the diagonal torus on the corresponding basis. They consider the composite alternative
\[
    \mathscr A_n=\{\rho^{\otimes n}\},\quad\mathscr B_n
    =\left\{
        g^{\otimes n}\sigma^{\otimes n}(g^\dagger)^{\otimes n}
        : g\in G\right\}.
\]
In the Stein limit, Corollary~26 in \cite{hayashi2026operational} shows that the optimal Type-II error exponent is given by $D_R(\rho\|\sigma)=S^{\mathrm{Keyl}}(\rho\|\sigma)$.

Keyl's rate appears naturally through pinching methods in quantum hypothesis testing, including the test achieving the optimal rate in Hayashi and Fang's work. Let
\(\mathcal P_{\rho^{\otimes n}}\) denote the pinching map onto the spectral
projections of \(\rho^{\otimes n}\).  If one pinches the alternative state
\(\sigma^{\otimes n}\) in the eigenbasis of the null state
\(\rho^{\otimes n}\), then
\[
    \lim_{n\to\infty}\frac1n D\left(\rho^{\otimes n}\,\middle\|\mathcal P_{\rho^{\otimes n}}(\sigma^{\otimes n})\right)=D_R(\rho\|\sigma)=S^{\mathrm{Keyl}}(\rho\|\sigma).
\]
This ``right-pinched'' or reverse-pinched relative entropy is the
\(\alpha\to1\) endpoint of the right-pinched Rényi divergences studied in \cite{lipka2024quantum}. In contrast, the standard pinching used in proofs of quantum Stein's lemma~\cite{hiai1991proper} pinches the first argument in the eigenbasis of the second and recovers the usual Umegaki relative entropy \(D(\rho\|\sigma)\). The same reverse-pinching mechanism appears in hypothesis testing with inconclusive outcomes. Ji and Regula~\cite{ji2025beyond} show that, for three-outcome tests with an abstention outcome, reverse pinching yields a test with conclusive probability tending to one while achieving the Umegaki exponent in one type of error and the Keyl exponent in the other.

An additional interpretation of Keyl's rate was given in \cite{notzel2015class} as the asymptotic rate function for a proposed quantum method of types. Keyl's rate also appeared, although implicitly in \cite{hayashi2025another}. There, Hayashi proposed a quantum empirical distribution, achieving Keyl's rate function in a version of Sanov's theorem for hypothesis testing. Keyl's tomography algorithm was generalized to arbitrary compact Lie groups via moment map estimation in \cite{botero2021large}, resulting in a variation of Keyl's rate function. 

\textbf{Open problems.} We end this subsection by recording several open directions suggested by the appearance of
\(S^{\mathrm{Keyl}}\) in tomography and hypothesis testing.

First, our covariant optimality theorem leaves open a regularity version of Keyl's original optimality question. Keyl conjectured that the lower semicontinuous envelope of admissible tomographic rate functions should be given by \(S^{\mathrm{Keyl}}\). Theorem~\ref{thm:relative-entropy-envelope-intro} shows that, without such regularity, the relative entropy \(D(\sigma\|\rho)\) is attainable pointwise by protocols tailored to \((\sigma,\rho)\). Moreover, Theorem~\ref{thm:main-introduction} shows that covariance forces the smaller bound \(S^{\mathrm{Keyl}}\).  It remains open whether joint
lower semicontinuity alone, without assuming covariance of the protocol, is already enough to imply
\[
    I(\sigma\|\rho)
    \le
    S^{\mathrm{Keyl}}(\sigma\|\rho)
    \qquad
    \forall \sigma,\rho\in\mathcal D_d .
\]

Second, it would be interesting to understand whether Keyl's rate admits a broader characterization through composite quantum hypothesis testing. In  simple binary hypothesis testing between two fixed i.i.d. states, the optimal Stein exponent is the Umegaki relative entropy. However, composite hypothesis testing is substantially more subtle: the optimal exponent can be governed by regularized entropic quantities and can be
strictly smaller than the naive worst-pairwise relative entropy
\cite{mosonyi2020error,berta2021composite,frenkel2025errorbounds,lami2025generalised}. This suggests the following problem: characterize the composite hypothesis testing problems whose optimal Stein exponent can achieve or be bounded by \(S^{\mathrm{Keyl}}\). 

Third, the relation between large-deviation rates and the usual sample complexity of tomography remains to be clarified. Sample complexity asks for
finite-\(n\) guarantees of the form
\[
    \Pr\{\|\sigma-\rho\|_1>\varepsilon\}\le\delta,
\]
whereas an LDP controls the exponential rate of this probability for fixed
\(\varepsilon\) and fixed dimension at sufficiently large $n$:
\[
    \Pr\{\|\sigma-\rho\|_1>\varepsilon\}
    \approx \exp\left\{-nI(\sigma\|\rho)\right\}.
\]
Thus, a pointwise larger rate function gives better asymptotic concentration, but this does not immediately imply better sample complexity uniformly in \(d\), rank, \(\varepsilon\), and \(\delta\), because subexponential prefactors, and
local small-\(\varepsilon\) behavior can matter. Recent works have established essentially optimal sample-complexity bounds for mixed-state tomography \cite{Haah_2017,o2016efficient,pelecanos2025mixed,hu2026sample}.
A natural problem is to determine how these finite-sample bounds are related to the rate functions. For example, does the
second-order expansion of \(I(\sigma\|\rho)\) near the diagonal
\(\sigma=\rho\) recover the optimal \(d\)- and rank-dependence of tomography? Conversely, can one turn a globally stronger LDP rate into uniform non-asymptotic sample-complexity improvements?

The remainder of the paper is organized as follows.  
Section~\ref{sec:prelim} recalls the representation-theoretic
background, including Schur--Weyl duality, weight decompositions, Schur polynomials, and the estimates used throughout the paper. 
Section~\ref{sec:general} develops the general POVM framework for tomography protocols and proves the structural characterization of covariant protocols. The large deviation principle and Keyl's conjecture are reviewed in Section~\ref{sec:LDP}. The proof of Theorem~\ref{thm:relative-entropy-envelope-intro} and Theorem~\ref{thm:main-introduction} are provided in Section~\ref{sec:relative-entropy} and Section~\ref{sec:main} respectively. Appendix~\ref{section:Appendix_KeylsRateProperties} collects a few properties of Keyl's rate function and Appendix~\ref{app:LDP} reviews a general contraction principle for LDP.

\subsection*{Acknowledgments}
 We thank B. Chen, V. Gandikota, J. Haah, P. Hayden, J. Wang, and H. Yang for helpful conversations.  The work of AM is supported in part by the Simons Foundation Grant No. 12574.  This work was  performed in part at the Aspen Center for Physics, which is supported by National Science Foundation grant PHY-2210452.



\subsection*{Notation}

Throughout the paper, \(d\in\mb N\) denotes the dimension of the underlying
complex Hilbert space \(\mb C^d\).

\begin{itemize}
    \item We denote by
    \[
        \mc D_d
        :=
        \{\rho\in M_d(\mb C): \rho\ge 0,\ \Tr(\rho)=1\}
    \]
    the set of quantum states on \(\mb C^d\). We also denote by
    \(\U(d)\) the group of unitary operators on \(\mb C^d\), by \(\GL(d)\) the group of invertible linear operators on \(\mb C^d\) and by $\mc L(\mb C^d)$ the set of all linear operators.

    \item Let
    \[
        \Delta_d
        :=
        \left\{
            x=(x_1,\ldots,x_d)\in\mb R^d:
            x_i\ge 0,\ \sum_{i=1}^d x_i=1
        \right\}
    \]
    be the probability simplex, and let
    \[
        \Delta_d^\downarrow
        :=
        \left\{
            x\in\Delta_d:
            x_1\ge x_2\ge \cdots \ge x_d
        \right\}
    \]
    be the set of probability vectors arranged in decreasing order. For any
    vector \(\alpha\in\mb R^d\), we write \(\alpha^\downarrow\) for the vector
    obtained by rearranging the entries of \(\alpha\) in decreasing order.

    \item For \(\alpha,\lambda\in\mb R^d\) with
    \(\sum_{i=1}^d \alpha_i=\sum_{i=1}^d \lambda_i\), we say that
    \(\alpha\) is majorized by \(\lambda\), denoted by
    \(\alpha\prec\lambda\), if
    \[
        \sum_{i=1}^k \alpha_i^\downarrow
        \le
        \sum_{i=1}^k \lambda_i^\downarrow,
        \qquad 1\le k\le d-1.
    \]

    \item For any matrix $A$, denote $\Delta_k(A)$ as the determinant of the leading \(k\times k\) principal minor of $A$: suppose $A = \sum_{i,j=1}^d a_{ij} |i\rangle \langle j|$, denote $A_{k\times k} = \sum_{i,j=1}^k a_{ij} |i\rangle \langle j|$ and 
\begin{equation}\label{eq:determinant-principle-minor}
    \Delta_k(A) = \det A_{k\times k}.
\end{equation}
\end{itemize}

\section{Preliminaries}\label{sec:prelim}
In this section, we gather the necessary representation-theoretic background. Standard references on representation theory include \cite{FultonHarris1991,Sagan2001,GoodmanWallach2009}. From the perspective of quantum information theory, see \cite{Christandl2006BipartiteStates,Leditzky2025RepTheoryQI,Harrow2005SchurTransform,hayashi2017, ODonnell2021}.
\subsection{Schur--Weyl duality}
Let \(d,n\in\mathbb N\), and write \([d]=\{1,\ldots,d\}\). 
We work with the \(n\)-partite Hilbert space
\[
    \mathcal H=(\mathbb C^d)^{\otimes n},
\]
equipped with its computational basis
\begin{equation}\label{eq:computational-basis}
    \bigl\{|i_1 i_2\cdots i_n\rangle : i_k\in [d]\bigr\}.
\end{equation}
Thus the notation \((\mathbb C^d)^{\otimes n}\) fixes both the tensor-product
structure and the computational bases. Let
\begin{equation}\label{eq:Young-diagram}
        \Lambda_{d,n}
    := \{\lambda=(\lambda_1,\ldots,\lambda_d)\in \mathbb Z_{\ge 0}^d:
    \lambda_1\ge \cdots \ge \lambda_d,\ 
    \sum_{i=1}^d \lambda_i=n\},
\end{equation}
which is the set of Young diagrams with \(n\) boxes and at most \(d\) rows. We denote $\lambda \vdash_d n$ as an element in $\Lambda_{d,n}$ and \(\ell(\lambda)\) as the number of nonzero
elements of \(\lambda\).

For each \(\lambda \vdash_d n\), let
\((\mathcal P_\lambda,p_\lambda)\) denote the irreducible \(S_n\)-representation
indexed by \(\lambda\), where \(\mathcal P_\lambda\) is the Specht module. Let
\((\mathcal Q_\lambda,q_\lambda)\) denote the irreducible polynomial
\(\GL(d)\)-representation\footnote{Alternatively, $U(d)$ representation.} of highest weight \(\lambda\), where
\(\mathcal Q_\lambda\) is the Schur module. Schur--Weyl duality decomposes the
tensor space as
\[
    \mathcal H
    =
    \bigoplus_{\lambda \vdash_d n}
    \mathcal P_\lambda\otimes \mathcal Q_\lambda .
\]
After choosing orthonormal bases of $\mc P_\lambda$ and $\mc Q_\lambda$, this
decomposition is implemented by a unitary change of coordinates, the
\textit{Schur transform},
\begin{equation}\label{eq:schur-transform-space}
    U_{\mathrm{Schur}}:
    (\mb C^d)^{\otimes n}
    \longrightarrow
    \bigoplus_{\lambda \vdash_d n}
    \mathcal P_\lambda\otimes \mathcal Q_\lambda .
\end{equation}
The corresponding orthonormal basis of $\bigoplus_{\lambda \vdash_d n}
    \mathcal P_\lambda\otimes \mathcal Q_\lambda $ is called a
\textit{Schur basis}.

The symmetric group \(S_n\) acts on \((\mb C^d)^{\otimes n}\) by permuting tensor factors:
\begin{equation}\label{eq:Sn-action}
    P(\pi)|i_1\cdots i_n\rangle
    :=
    |i_{\pi^{-1}(1)}\cdots i_{\pi^{-1}(n)}\rangle,
    \qquad \pi\in S_n.
\end{equation}
For an $X \in \mathrm{GL}(d)$, $X^{\otimes n}$ defines an action on  on \((\mb C^d)^{\otimes n}\). 
These two
actions commute. In the Schur basis, we have 
\begin{equation}\label{eq:schur-combined-action}
    U_{\mathrm{Schur}}
    \bigl(P(\pi)X^{\otimes n}\bigr)
    U_{\mathrm{Schur}}^\dagger
    =
    \bigoplus_{\lambda \vdash_d n}
    p_\lambda(\pi)\otimes q_\lambda(X).
\end{equation}
If one represents the direct sum using an explicit block-label register, then
\eqref{eq:schur-combined-action} may equivalently be written as
\[
    U_{\mathrm{Schur}}
    \bigl(P(\pi)X^{\otimes n}\bigr)
    U_{\mathrm{Schur}}^\dagger
    =
    \sum_{\lambda \vdash_d n}
    |\lambda\rangle\langle\lambda|
    \otimes p_\lambda(\pi)\otimes q_\lambda(X).
\]
Since \(q_\lambda\) is a polynomial representation, it extends uniquely from
\(\mathrm{GL}(d)\) to all \(X\in \mc L(\mathbb C^d)\). Via Schur's lemma, if $X^n \in \mc L((\mb C^d)^{\otimes n})$ is permutation invariant, i.e.,
\[
P(\pi) X^n  P(\pi)^\dagger= X^n,\quad \forall \pi \in S_n,
\]
then for any $\lambda \vdash_d n$, there exists $X^n_\lambda \in \mc L(\mc Q_\lambda)$ such that 
\begin{equation}\label{eq:permutation-invariant-general}
    U_{\mathrm{Schur}} X^n\, U_{\mathrm{Schur}}^\dagger = 
    \sum_{\lambda \vdash_d n}
    |\lambda\rangle\langle\lambda|
    \otimes \textbf{1}_{\mc P_\lambda}\otimes X^n_\lambda.
\end{equation}

\subsection{Weights of representations of general linear groups}
For \(\alpha\in\mathbb Z_{\ge 0}^d\) with \(|\alpha|=n\), we use the conventions \(0^0=1\) and \(0^m=0\) for \(m>0\) and denote the monomial
\[
t^\alpha:= \prod_{k=1}^d t_k^{\alpha_k},\quad t \in \mb C^d.
\]
Define the \(\alpha\)-weight space of \(\mathcal Q_\lambda\) by
\begin{equation}\label{eq:weight-space-definition}
        \mathcal Q_{\lambda,\alpha}
    :=
    \left\{
        v\in\mathcal Q_\lambda:
        q_\lambda(\diag(t))v=t^\alpha v
        \ \text{for all } t \in(\mathbb C \backslash \{0\})^d
    \right\}. 
\end{equation}
Equivalently, it is enough to test this identity on the compact diagonal torus
\[
    \mathrm{T}(d)
    :=
    \{\diag(z_1,\ldots,z_d): |z_i|=1\}.
\]
The weight-space decomposition is
\[
    \mathcal Q_\lambda
    =
    \bigoplus_{\alpha\in \mathrm{Wt}_d(\lambda)}
    \mathcal Q_{\lambda,\alpha},
\]
where the weights of a Young diagram $\lambda$ are given by 
\[
    \mathrm{Wt}_d(\lambda)
    :=
    \{\alpha\in\mathbb Z_{\ge 0}^d: |\alpha|=n,\ 
      \mathcal Q_{\lambda,\alpha}\neq 0\} =  \{\alpha\in\mathbb Z_{\ge 0}^d: |\alpha|=n,\ 
      \alpha \prec \lambda\}.
\]
The multiplicity of the weight \(\alpha\) is $\dim \mathcal Q_{\lambda,\alpha}$, and is known as a \textit{Kostka number}.

Let \(P_{\lambda,\alpha}\) denote the orthogonal projection from
\(\mathcal Q_\lambda\) onto \(\mathcal Q_{\lambda,\alpha}\). The following well-known facts will be used repeatedly:

\begin{enumerate}
    \item \textbf{Diagonal action.}
    For \(\diag(t)=\diag(t_1,\ldots,t_d)\) with \(t_i>0\),
    \begin{equation}\label{eq:diagonal-decomposition}
        q_\lambda(\diag(t))
        =
        \sum_{\alpha\in\mathrm{Wt}_d(\lambda)}
        t^\alpha P_{\lambda,\alpha}.
    \end{equation}
    Since \(q_\lambda\) is polynomial, the same identity extends by continuity
    to \(t_i\ge 0\).
    \medskip
    \item \textbf{Fourier formula for weight projections.}
    For \(\alpha\in\mathrm{Wt}_d(\lambda)\), and $U(\theta) = \diag(e^{i\theta_1},\ldots,e^{i\theta_d})$,
    \begin{equation}\label{eq:torus-Fourier}
    P_{\lambda,\alpha}
        =
        \int_{[0,2\pi]^d}
        e^{-i\alpha\cdot\theta}
        q_\lambda(U(\theta))
        \frac{d\theta}{(2\pi)^d}.
\end{equation}
    If \(\alpha\notin\mathrm{Wt}_d(\lambda)\), the same integral is \(0\).
    \medskip
    \item \textbf{Torus twirling kills off-diagonal weight blocks.}
    For \(K_\lambda\in\mc L(\mathcal Q_\lambda)\), and $U(\theta) = \diag(e^{i\theta_1},\ldots,e^{i\theta_d})$,
    \begin{equation}\label{eq:twirl-torus}
   \int_{[0,2\pi]^d}
        q_\lambda(U(\theta))K_\lambda\,q_\lambda(U(\theta))^\dagger
        \frac{d\theta}{(2\pi)^d}
        =
        \sum_{\alpha\in\mathrm{Wt}_d(\lambda)}
        P_{\lambda,\alpha}K_\lambda P_{\lambda,\alpha}.
\end{equation}
    \item \textbf{Upper triangular matrices can only increase weights.}
    Let $L$ be an upper triangular matrix, then
    \begin{equation}\label{eq:triangular-weight}
     q_\lambda(L)\mathcal Q_{\lambda,\alpha}
        \subseteq \bigoplus_{\substack{\beta\in\mathrm{Wt}_d(\lambda)\\ \alpha \prec \beta}}
        \mathcal Q_{\lambda,\beta}.
\end{equation}
    To see the direction of the order, note that the elementary matrix $|i\rangle \langle j |,\ i < j$ changes a weight \(\alpha\) to \(\alpha+e_i-e_j\), which is larger in dominance order. Upper triangular matrices are generated by the diagonal torus and such raising operations.
\end{enumerate}
We write \(\mu_{\Haar}\) for the normalized Haar probability measure on the unitary group \(\U(d)\), namely the unique Borel probability measure satisfying $ \mu_{\Haar}(VU)=\mu_{\Haar}(U)=\mu_{\Haar}(UV)$ for fixed $U\in \U(d)$ and all \(V\in \U(d)\). The following property will be used repeatedly.
\medskip
\begin{lemma}[Schur orthogonality]\label{lemma:Schur-orthogonality}
Fix $\lambda \vdash_d n$. For every \(K_\lambda\in \mc L(\mathcal Q_\lambda)\),
\begin{equation}
    \int_{\mathcal U(d)}
    q_\lambda(U)\,K_\lambda\,q_\lambda(U)^\dagger\, d\mu_{\Haar}(U)
    =
    \frac{\Tr(K_\lambda)}{\dim \mathcal Q_\lambda}\,
    \mathbf 1_{\mathcal Q_\lambda}.
\end{equation}
\end{lemma}

\subsection{Schur polynomials and concentration of Young diagrams}

Since \(q_\lambda\) is a polynomial representation, the map
\(X\mapsto q_\lambda(X)\) extends from \(\GL(d)\) to all
\(X\in \mc L(\mathbb C^d)\). We define the Schur polynomial, or Schur character,
by
\begin{equation}
    s_\lambda(X):=\Tr q_\lambda(X).
\end{equation}
When \(X\ge 0\) has eigenvalues $r = (r_1,\cdots,r_d)$ with $r_1\ge r_2\ge\cdots\ge r_d\ge 0$, we also write
\[
    s_\lambda(\diag(r))=s_\lambda(X).
\]

We will use the following standard estimates repeatedly.
\medskip
\begin{lemma}[Representation-theoretic estimates]\label{lemma:rep-data}
For \(\lambda\vdash_d n\), let \(\bar\lambda:=\lambda/n\). Then
\begin{equation}\label{eq:dimension-estimates}
    \dim \mathcal Q_\lambda
    \le
    (n+1)^{d(d-1)/2},
\end{equation}
and denote the entropy $H(p):=-\sum_{i=1}^d p_i\log p_i,\ p \in \Delta_d$, we have
\begin{equation}\label{eq:specht-dimension-estimates}
    (n+d)^{-d(d+1)/2} e^{nH(\bar\lambda)}
    \le
    \dim \mathcal P_\lambda
    \le
    e^{nH(\bar\lambda)}.
\end{equation}
Moreover, if \(X\ge 0\) has eigenvalues
$r = (r_1,\cdots,r_d)$ with $r_1\ge r_2\ge\cdots\ge r_d\ge 0$, then
\begin{equation}\label{eq:schur-polynomial-upper-lower}
    r^\lambda
    \le
    s_\lambda(X)
    \le
    \dim\mathcal Q_\lambda\,
    r^\lambda.
\end{equation}
\end{lemma}
Let \(\rho\in\mathcal D_d\). In Schur basis, we have
\[
    U_{\mathrm{Schur}}\rho^{\otimes n}\,U_{\mathrm{Schur}}^\dagger
    =
    \sum_{\lambda\vdash_d n} |\lambda\rangle \langle \lambda|\otimes 
    \mathbf 1_{\mathcal P_\lambda}\otimes q_\lambda(\rho).
\]
Denote the projection onto the \(\lambda\)-block as
\[
    \Pi_\lambda
    =
    |\lambda\rangle\langle\lambda|
    \otimes
    \mathbf 1_{\mathcal P_\lambda}
    \otimes
    \mathbf 1_{\mathcal Q_\lambda}.
\]
Performing the measurements $\{\wt \Pi_\lambda =U_{\mathrm{Schur}}^\dagger\Pi_\lambda\,U_{\mathrm{Schur}}\}_{\lambda \vdash_d n}$ defines a probability distribution on
\(\Lambda_{d,n}\):
\begin{equation}\label{eq:probability-Young}
    \mathbb P_{n,\rho}(\lambda)
    :=
    \Tr(\wt \Pi_\lambda\rho^{\otimes n})
    =
    \dim\mathcal P_\lambda\, s_\lambda(\rho).
\end{equation}
For \(p,x\in\Delta_d\), define the relative entropy
$D(p\|x)
    :=
    \sum_{i=1}^d p_i\log\frac{p_i}{x_i}$, with the convention that \(D(p\|x)=+\infty\) if
\(p_i>0\) for some \(i\) with \(x_i=0\). A well-known Keyl--Werner spectrum-estimation bound \cite{KeylWerner2001, HayashiMatsumoto2002}, see also \cite[Theorem~1 and Corollary~1]{ChristandlMitchison2006}, is given by
\begin{equation}\label{eq:keyl-werner-single}
    \mathbb P_{n,\rho}(\lambda)
    \le
    (n+1)^{d(d-1)/2}
    \exp\!\left(
        -nD\!\left(\overline{\lambda} \|x\right)
    \right).
\end{equation}
As an application, we have the following concentration property:
\medskip
\begin{lemma}\label{lemma:Schur-concentration}
    For any state $\rho \in \mc D_d$ with $x = \mathrm{spec}^\downarrow(\rho)$ and $\delta> 0$, denote 
    \begin{equation}
        G_{n,\delta}(x):= \{\lambda: \|\lambda/n  -x\|_1 < \delta, \ \ell(\lambda) \le \rank(x)\}.
    \end{equation}
Then 
    \begin{equation}
       \lim_{n\to \infty} \mb P_{n,\rho}(G_{n,\delta}(x)) = 1.
    \end{equation}
\end{lemma}

\section{General framework for quantum state tomography}\label{sec:general}
Let $\mathcal H_n=(\mathbb C^d)^{\otimes n}$ be the $n$-fold tensor product Hilbert space, equipped with its computational basis. Given $n$ copies of an unknown state $\rho$, a tomography protocol produces an estimator of $\rho$. 
Thus a tomography protocol is naturally a POVM whose outcome space is the state space \(\mathcal D_d\). 

To allow estimators with a continuum of possible values, we use Borel subsets of \(\mathcal D_d\), denoted as $\mathfrak B(\mathcal D_d)$. 
\medskip
\begin{definition}\label{def:tomography-general}
    For any $n \ge 1$, a quantum state tomography protocol with output space \(\mathcal D_d\), is a $\sigma$-additive positive-operator-valued measure
    \begin{equation}
        E_n:\mathfrak B(\mathcal D_d)\to \mc L(\mathcal H_n)_+,\quad E_n(\mathcal D_d)=\mathbf 1_{\mathcal H_n},\quad E_n(\emptyset) = 0.
    \end{equation}
\end{definition}
Given a protocol \(\{E_n\}_{n\ge 1}\) and \(n\) copies of an unknown state
\(\rho\), the distribution of the estimator is determined by the Born rule:
\begin{equation}\label{eq:distribution-estimator}
    \widehat\mu^\rho_n(A):=\Tr\bigl(E_n(A)\rho^{\otimes n}\bigr).
\end{equation}
It is often convenient to write the POVM \(E_n\) in a density form. Since \(\mathcal H_n\) is finite-dimensional, every such POVM admits a density with respect to some probability measure on \(\mathcal D_d\). Indeed, one may
take
\begin{equation}\label{eq:prob-measure-generic}
    \mu_n(A):=\frac{1}{\dim\mathcal H_n}\Tr E_n(A).
\end{equation}
Now, if \(\mu_n(A)=0\) 
then \(E_n(A)=0\). Hence \(E_n\) is absolutely continuous with respect to \(\mu_n\).
Applying the scalar Radon--Nikodym theorem to the matrix entries of \(E_n\)
gives a measurable positive-operator-valued density
\[
    \sigma\in\mathcal D_d
    \longmapsto
    M_\sigma^{(n)}\in \mc L(\mathcal H_n)_+
\]
such that
\[
    E_n(A)=\int_A M_\sigma^{(n)}\,d\mu_n(\sigma),
    \qquad
    \int_{\mathcal D_d}M_\sigma^{(n)}\,d\mu_n(\sigma)
    =
    \mathbf 1_{\mathcal H_n}.
\]
We call $M_\sigma^{(n)}$ a POVM \textit{density}. With this representation, the probability density of obtaining the estimate \(\sigma\), when the true state is \(\rho\), is
\[
    P_n^M(\sigma|\rho)
    :=
    \Tr\bigl(M_\sigma^{(n)}\rho^{\otimes n}\bigr),
\]
and the probability that the estimator belongs to a subset \(A\subseteq\mathcal D_d\)
is
\[
    \widehat \mu^\rho_n(A)
    =
    \int_A P_n^M(\sigma|\rho)\,d\mu_n(\sigma).
\]
What this means is that, rather than starting from Definition~\ref{def:tomography-general}, we can alternatively characterize a general quantum state tomography protocol as a sequence of measurable families of positive operators \[
\sigma\in\mathcal D_d
    \longmapsto
    M_\sigma^{(n)}\in \mc L(\mathcal H_n)_+
\]
equipped with a sequence of probability measures $\mu_n$ on $\mc D_d$, and subject to the usual constraint
$\int_{\mathcal D_d}M_\sigma^{(n)}\,d\mu_n(\sigma)
    =
    \mathbf 1_{\mathcal H_n}$.
We denote such a tomography protocol as 
\begin{equation}
    \mathfrak T_{\mu,M} = \left(\{M_\sigma^{(n)}\}_{\sigma \in \mc D_d},\mu_n \right)_{n\ge 1}.
\end{equation}
This representation is not unique: different choices of \(\mu_n\) and \(M_\sigma^{(n)}\) may define the same POVM \(E_n\). Moreover, different POVMs may induce the same distribution of estimators. We say two POVMs $\{E_n\}_{n\ge 1},\ \{E'_n\}_{n\ge 1}$ are \textit{i.i.d.\ statistically equivalent} if, for any $\rho$ and $n\ge 1$, we have 
\[
\Tr(E_n(A) \rho^{\otimes n}) = \Tr(E'_n(A) \rho^{\otimes n}),\quad \forall A \in \mathfrak B(\mc D_d).
\]

\subsection*{Covariant protocols} 
We will focus our attention on \emph{covariant} protocols. 
Intuitively, a tomography protocol is covariant if rotating the true state by a unitary \(U\) simply rotates the output estimate
by the same unitary.  This will provide substantially more structure on the space of tomography protocols.
\medskip
\begin{definition}
    We say that a quantum state tomography protocol $\{E_n\}_{n\ge 1}$ is \emph{covariant}, if for any $A \in \mathfrak B(\mc D_d)$, $U \in \U(d)$,
\[
E_n(U A U^\dagger) = U^{\otimes n} E_n(A) (U^\dagger)^{\otimes n},\quad UAU^\dagger
    :=
    \{U\sigma U^\dagger:\sigma\in A\}.
\]
\end{definition}
If \(E_n\) has a density representation
\[
    E_n(A)=\int_A M_\sigma^{(n)}\,d\mu_n(\sigma),
\]
then covariance is equivalently expressed as
\begin{equation}\label{eq:covariant-density}
    \int_{UAU^\dagger} M_\sigma^{(n)}\,d\mu_n(\sigma)
    =
    \int_A
    U^{\otimes n}M_\sigma^{(n)}(U^\dagger)^{\otimes n}
    \,d\mu_n(\sigma).
\end{equation}
In general, the density \(M_\sigma^{(n)}\) itself need not satisfy a pointwise
covariance relation. The next lemma shows that, for covariant protocols, one
can choose a density realization for which such a relation holds.
\medskip
\begin{lemma}\label{lem:covariance-density}
    Suppose the protocol $\{E_n\}_{n\ge 1}$ is covariant, then there exist unitarily invariant probability measures $\mu_n$ on $\mc D_d$ and measurable familty \(\{M_\sigma^{(n)}\}_{\sigma \in \mc D_d}\) with 
    \[
    E_n(A) = \int_A M_\sigma^{(n)} d\mu_n(\sigma), \quad \forall A \in \mathfrak B(\mc D_d),
    \]
    and $\mu_n$-almost surely,
    \[
    M_{U\sigma U^\dagger}^{(n)} = U^{\otimes n} M_\sigma^{(n)} (U^\dagger)^{\otimes n},\quad \forall U\in \U(d).
    \]
\end{lemma}
\begin{proof}
    Since $E_n$ is covariant, we have that
    \[
    \mu_n(\cdot):= \frac{1}{d^n} \Tr(E_n(\cdot))
    \]
    is unitarily invariant. We denote $\wt M_\sigma^{(n)}$ as the POVM density with respect to $\mu_n$, i.e., 
    \[
    E_n(A) = \int_A \wt M_\sigma^{(n)} d\mu_n(\sigma).
    \]
    Define
    \begin{align*}
        M_\sigma^{(n)} := \int_{\U(d)} (W^\dagger)^{\otimes n} \wt M_{W \sigma W^\dagger}^{(n)} W^{\otimes n} d\mu_{\Haar}(W).
    \end{align*}
    Via the left invariance of the Haar measure, we have, $\mu_n$-almost surely, that $M_{U\sigma U^\dagger}^{(n)} = U^{\otimes n} M_\sigma^{(n)} (U^\dagger)^{\otimes n}$ for all $U\in \U(d)$ and 
    \begin{align*}
        \int_A M_\sigma^{(n)} d\mu_n(\sigma) & = \int_A \int_{\U(d)} (W^\dagger)^{\otimes n} \wt M_{W \sigma W^\dagger}^{(n)} W^{\otimes n} d\mu_{\Haar}(W) d\mu_n(\sigma) \\
        & = \int_{\U(d)} (W^\dagger)^{\otimes n} \int_A \wt M_{W \sigma W^\dagger}^{(n)}  d\mu_n(\sigma) W^{\otimes n} d\mu_{\Haar}(W) \\
        & = \int_{\U(d)} (W^\dagger)^{\otimes n} E_n(WAW^\dagger) W^{\otimes n} d\mu_{\Haar}(W) \\
        & = E_n(A),
    \end{align*}
    where the first equality is the definition of $M_\sigma^{(n)}$; the second equality follows from exchanging the integral; the third equality uses unitary invariance of $\mu_n$ and the last equality follows from the covariance of $E_n$. Thus $M_\sigma^{(n)}$ is also a POVM density for $E_n$, concluding the proof.
\end{proof}
Recall that two protocols are i.i.d.\ statistically equivalent if
they induce the same distribution of the estimators for every input of the form \(\rho^{\otimes n}\). A simple observation is that every protocol is statistically equivalent to a permutation invariant protocol. Thus we have the following characterization:
\medskip
\begin{lemma}\label{lem:permutation-invariance-density}
    Suppose the protocol $\{E_n\}_{n\ge 1}$ is covariant, then there exist unitarily invariant probability measures $\mu_n$ on $\mc D_d$ and measurable familty \(\{M_\sigma^{(n)}\}_{\sigma \in \mc D_d}\) such that $\mu_n$-almost surely,
    \[
    P_\pi M_\sigma^{(n)} = M_\sigma^{(n)} P_\pi,\quad M_{U\sigma U^\dagger}^{(n)} = U^{\otimes n} M_\sigma^{(n)} (U^\dagger)^{\otimes n},\quad \forall U\in \U(d),\ \pi\in S_n,
    \]
and for any $A \in \mathfrak B(\mc D_d)$ and $\rho \in \mc D_d$,
\begin{equation}\label{eq:POVM-same-outcome}
    \Tr(E_n(A) \rho^{\otimes n}) = \int_A \Tr(M_\sigma^{(n)} \rho^{\otimes n})d\mu_n(\sigma).
\end{equation}
\end{lemma}
\begin{proof}
    Using Lemma~\ref{lem:covariance-density}, there exist unitarily invariant $\mu_n$ and $\{\wt M_\sigma^{(n)}\}_{\sigma \in \mc D_d}$ such that $\mu_n$-almost surely,
    \[
    \wt M_{U\sigma U^\dagger}^{(n)} = U^{\otimes n} \wt M_\sigma^{(n)} (U^\dagger)^{\otimes n},\quad \forall U\in \U(d).
    \]
    Now define 
    \[
    M_\sigma^{(n)} = \frac{1}{n!} \sum_{\pi \in S_n} P_\pi \wt M_\sigma^{(n)} P_\pi^\dagger.
    \]
    Using the fact that $P(\pi)$ and $U^{\otimes n}$ commute, we have $\mu_n$-almost surely,
    \[
    P_\pi M_\sigma^{(n)} = M_\sigma^{(n)} P_\pi,\quad M_{U\sigma U^\dagger}^{(n)} = U^{\otimes n} M_\sigma^{(n)} (U^\dagger)^{\otimes n},\quad \forall U\in \U(d),\ \pi\in S_n.
    \]
    Moreover, since $\rho^{\otimes n}$ is permutation invariant, we have $\Tr(E_n(A) \rho^{\otimes n}) = \int_A \Tr(M_\sigma^{(n)} \rho^{\otimes n})d\mu_n(\sigma)$, concluding the proof.
\end{proof}
We next record the form of a unitarily invariant probability measure $\mu_n$ on $\mc D_d$. Let $$s: \mc D_d\to \Delta_d^\downarrow,\ s(\sigma) := \mathrm{spec}^\downarrow(\sigma)$$ be the spectrum map, \(\nu_n\) is defined as the probability measure on $\Delta_d^\downarrow$:
\begin{equation}\label{eq:measure-on-simplex}
    \nu_n(\Sigma) = \mu_n(s^{-1}(\Sigma)),\quad \forall \Sigma \subseteq \Delta_d^\downarrow.
\end{equation}
Then \(\mu_n\) can be reconstructed by first sampling the spectrum
\(x\sim\nu_n\), then sampling \(U\sim\mu_{\Haar}\), and finally setting
\(\sigma=U\diag(x)U^\dagger\). 
Equivalently, for every measurable integrable function \(F: \mc D_d \to \mb R\),
\begin{equation}\label{eq:mu-disintegration}
    \int_{\mc D_d} F(\sigma)\,d\mu_n(\sigma)
    =
    \int_{\Delta^\downarrow_d}\int_{\U(d)}
    F(U \diag(x)U^\dagger)\,
    d\mu_{\Haar}(U)\,d\nu_n(x).
\end{equation}
Finally, we present a characterization of covariant protocols up to statistical equivalence:
\medskip
\begin{theorem}\label{thm:characterization-covariant}
Suppose the protocol $\{E_n\}_{n\ge 1}$ is covariant, then there exist unitarily invariant probability measures $\mu_n$ on $\mc D_d$ with decomposition~\eqref{eq:mu-disintegration} and POVM density $\{M_\sigma^{(n)}\}_{\sigma \in \mc D_d}$ of the form
    \begin{equation}
        M_{U\diag(x)U^\dagger}^{(n)}
    =
    U_{\mathrm{Schur}}^\dagger
    \left(
        \sum_{\lambda\vdash_d n} \dim \mc Q_\lambda
        |\lambda\rangle\langle\lambda|
        \otimes \mathbf 1_{\mc P_\lambda}
        \otimes
        q_\lambda(U)K_{\lambda,x}q_\lambda(U^\dagger)
    \right)
    U_{\mathrm{Schur}},
    \end{equation}
such that for any $A \in \mathfrak B(\mc D_d)$ and $\rho \in \mc D_d$, we have
\[
\Tr(E_n(A) \rho^{\otimes n}) = \int_{A} \Tr(M_\sigma^{(n)} \rho^{\otimes n})d\mu_n(\sigma).
\]
For each $\lambda \vdash_d n$ and $x \in \Delta_d^\downarrow$, we have $K_{\lambda,x} \in \mc L(\mc Q_\lambda)_+$ and 
\begin{equation}
    \int_{\Delta_d^\downarrow}
    \Tr K_{\lambda,x}d\nu_n(x) = 1,\quad \forall \lambda \vdash_d n,
\end{equation}
where $\nu_n$ is defined by \eqref{eq:measure-on-simplex}. 
\end{theorem}
\begin{proof}
 By Lemma~\ref{lem:permutation-invariance-density}, we may replace \(E_n\),
    without changing its statistics on i.i.d.\ inputs, by a protocol with a
    density \(M_\sigma^{(n)}\) satisfying both permutation invariance and
    unitary covariance:
    \[
    P_\pi M_\sigma^{(n)} = M_\sigma^{(n)} P_\pi,\quad M_{U\sigma U^\dagger}^{(n)} = U^{\otimes n} M_\sigma^{(n)} (U^\dagger)^{\otimes n},\quad \forall U\in \U(d),\ \pi\in S_n,
    \]
Recall that if $\{M_\sigma^{(n)}\}_{\sigma \in \mc D_d}$ is permutation invariant, it can be characterized as~\eqref{eq:permutation-invariant-general}
\[
M_\sigma^{(n)} =  U_{\mathrm{Schur}}^\dagger
    \left(
        \sum_{\lambda\vdash_d n}
        |\lambda\rangle\langle\lambda|
        \otimes \mathbf 1_{\mc P_\lambda}
        \otimes K_\lambda(\sigma)
    \right)
    U_{\mathrm{Schur}},
\]
where $\{K_\lambda(\sigma)\}_{\sigma \in \mc D_d}$ is a measurable family acting on $\mc Q_\lambda$ for each $\lambda \vdash_d n$. 

The covariance of \(M_\sigma^{(n)}\), together with the Schur--Weyl action
    of \(U^{\otimes n}\), gives
    \begin{equation}\label{eq:block-density-covariance}
        K_\lambda(U\sigma U^\dagger)
        =
        q_\lambda(U)K_\lambda(\sigma)q_\lambda(U^\dagger).
    \end{equation}
    Define
\[
    K_{\lambda,x}
    :=
    \frac{1}{\dim \mc Q_{\lambda}} \, K_\lambda(\diag(x)),\quad \lambda\vdash_d n,\ x \in \Delta_d^\downarrow.
\]
Then, for \(\sigma=U\diag(x)U^\dagger\),
\begin{equation}\label{eq:covariant-K-decomposition}
    K_\lambda(\sigma)
    =
    \dim \mc Q_{\lambda} q_\lambda(U)K_{\lambda,x}q_\lambda(U^\dagger).
\end{equation}
Consequently,
\begin{equation}\label{eq:general-covariant-density-K}
\begin{aligned}
    M_{U\diag(x)U^\dagger}^{(n)}
    =
    U_{\mathrm{Schur}}^\dagger
    \left(
        \sum_{\lambda\vdash_d n} \dim \mc Q_{\lambda}
        |\lambda\rangle\langle\lambda|
        \otimes \mathbf 1_{\mc P_\lambda}
        \otimes
        q_\lambda(U)K_{\lambda,x}q_\lambda(U^\dagger)
    \right)
    U_{\mathrm{Schur}} .
\end{aligned}
\end{equation}
 It remains to identify the normalization condition. Indeed, comparing the \(\lambda\)-blocks in Schur basis gives
\begin{equation}\label{eq:K-lambda-x-normalization-twirl}
   \int_{\mc D_d} K_\lambda(\sigma) d\mu_n(\sigma) = \dim \mc Q_{\lambda}\int_{\Delta_d^\downarrow}\int_{\U(d)}
    q_\lambda(U)K_{\lambda,x}\,q_\lambda(U^\dagger)\,
    d\mu_{\Haar}(U)\,d\nu_n(x)
    =
    \mathbf 1_{\mc Q_\lambda},
    \qquad
    \forall \lambda\vdash_d n.
\end{equation}
Equivalently, by Schur orthogonality, see Lemma~\ref{lemma:Schur-orthogonality}, we have
\begin{equation}\label{eq:K-lambda-x-trace-normalization}
    \int_{\Delta_d^\downarrow}
    \Tr K_{\lambda,x}d\nu_n(x)
    = 1,
    \qquad
    \forall \lambda\vdash_d n.
\end{equation}
\end{proof}
In summary, a general covariant tomography protocol is i.i.d.\ statistically equivalent to the following protocol:
\begin{tcolorbox}[
    colback=white,
    colframe=black,
    boxrule=0.8pt,
    arc=1pt,
    left=6pt,
    right=6pt,
    top=6pt,
    bottom=6pt
]
Given \(n\)-copies of \(\rho\):
\begin{enumerate}
    \item First apply the Schur transform to \(\rho^{\otimes n}\).
    \item In Schur basis, perform the POVM of the form 
    \[\sum_{\lambda\vdash_d n} \dim \mc Q_{\lambda}
        |\lambda\rangle\langle\lambda|
        \otimes \mathbf 1_{\mc P_\lambda}
        \otimes
        q_\lambda(U)K_{\lambda,x}q_\lambda(U^\dagger)\]
    where $x \sim \nu_n$, $U \sim \mu_{\Haar}$ and $\int_{\Delta_d^\downarrow}
    \Tr K_{\lambda,x}d\nu_n(x)
    = 1$.
    \item Collect the measurement outcomes \(\widehat{x}\) and \(\widehat{U}\), and output an estimator
    \[
        \sigma
        =
        \widehat{U}\operatorname{diag}(\widehat{x})\widehat{U}^{\dagger}.
    \]
\end{enumerate}
\end{tcolorbox}

\section{Large deviation principle and Keyl's conjecture}\label{sec:LDP}
In this section, we formally introduce the notion of large deviation principle, and review the conjecture proposed in \cite{Keyl_2006}. First, we introduce the notion of \textit{consistency}. Recall that any quantum state tomography protocol induces the estimator distribution
\begin{equation}\label{eq:probability-tomography-general}
    \widehat\mu^{\rho}_n(A)= \Tr(E_n(A) \rho^{\otimes n}) = \int_A \Tr(M_\sigma^{(n)} \rho^{\otimes n})d\mu_n(\sigma).
\end{equation}
A good tomography protocol should concentrate its estimator near the true state. This motivates the following notion:
\medskip
\begin{definition}\label{def:consistency}
    We say that the protocol $\left(\{M_\sigma^{(n)}\}_{\sigma \in \mc D_d},\mu_n \right)_{n\ge 1}$ is consistent, if for any $\varepsilon >0$ and $\rho \in \mc D_d$,
    \begin{equation}\label{eq:neighborhood}
        \lim_{n\to\infty} \widehat\mu^{\rho}_n\left( N_\varepsilon(\rho)\right) = 1,\quad N_\varepsilon(\rho):= \{\tau \in \mc D_d: \|\tau - \rho\|_1 < \varepsilon\},
    \end{equation}
    where $\widehat\mu^{\rho}_n(\cdot)$ is defined by~\eqref{eq:probability-tomography-general}. 
\end{definition}
In other words, consistency is the ordinary probabilistic requirement that the estimator converges in probability to the true density matrix, which is equivalent to finiteness of sample complexity.

The following formal definition from \cite[Definition III.6]{hollander2000large} characterizes the exponential rate at which the protocol assigns probability to a particular estimate $\sigma$ of the true state $\rho$, informally introduced in \eqref{rate}. 
\medskip
\begin{definition}\label{def:TomRateFunc}
    We say that the protocol $\left(\{M_\sigma^{(n)}\}_{\sigma \in \mc D_d},\mu_n \right)_{n\ge 1}$ satisfies a large deviation principle (LDP) with rate function $I(\cdot \| \rho): \mc D_d \to [0,\infty]$ for any $\rho \in \mc D_d$, if $I(\cdot \| \rho)$ satisfies
    \begin{enumerate}
        \item $I(\cdot \| \rho) \not\equiv \infty$; 
        \item $I(\cdot \| \rho)$ is lower semicontinuous: for any sequence $\sigma_n \to \sigma$ $\liminf_{n\to \infty}I(\sigma_n \| \rho) \ge I(\sigma\| \rho)$;
        \item $\{\sigma \in \mc D_d: I(\sigma \| \rho) \in [0,c]\}$ is compact for any $c \ge 0$,
    \end{enumerate}
    and we have 
    \begin{itemize}
        \item for any open set $O \subseteq \mc D_d$, 
        $$\limsup_{n\to \infty} -\frac{1}{n}\log \widehat\mu^{\rho}_n(O) \le \inf_{\sigma \in O} I(\sigma \| \rho).$$
        \item For any closed set $F \subseteq \mc D_d$,
        $$\liminf_{n\to \infty} -\frac{1}{n}\log \widehat\mu^{\rho}_n(F) \ge \inf_{\sigma \in F} I(\sigma \| \rho).$$
    \end{itemize}
\end{definition}
We emphasize that our notion of consistency is purely probabilistic: the estimator distribution must converge to a point mass at the true state. This is weaker than requiring the LDP rate function to vanish only at the true state.
Indeed, in \cite[Definition 2.2]{Keyl_2006}, Keyl includes the following unique-zero condition in the definition of
an LDP rate function for an estimation scheme:
\begin{equation}\label{eq:unique-zero}
    I(\sigma\|\rho)=0
    \quad\Longleftrightarrow\quad
    \sigma=\rho .
\end{equation}
We show that this condition, together with an LDP, implies consistency:
\medskip
\begin{proposition}
\label{lemma:unique-zero-implies-consistency}
Suppose that the protocol satisfies an LDP with rate function
\(I(\cdot\|\rho)\). Assume that, for every \(\rho\in\mathcal D_d\),
\[
    I(\sigma\|\rho)=0
    \quad\Longleftrightarrow\quad
    \sigma=\rho.
\]
Then the protocol is consistent.
\end{proposition}
\begin{proof}
Fix \(\rho\in\mathcal D_d\) and \(\varepsilon>0\). Let
\[
    F_\varepsilon(\rho)
    :=
    \{\sigma\in\mathcal D_d:\|\sigma-\rho\|_1\ge\varepsilon\}.
\]
This set is compact and does not contain \(\rho\). Since
\(I(\cdot\|\rho)\) is lower semicontinuous, it attains its minimum on
\(F_\varepsilon(\rho)\). By the unique-zero assumption, this minimum cannot be
zero. Hence
\[
    \inf_{\sigma\in F_\varepsilon(\rho)}
    I(\sigma\|\rho)>0.
\]
The closed-set LDP bound gives
\[
    \liminf_{n\to\infty}
    -\frac1n
    \log
    \widehat\mu_n^{\rho}
    \bigl(F_\varepsilon(\rho)\bigr)
    \ge
    \inf_{\sigma\in F_\varepsilon(\rho)}
    I(\sigma\|\rho) > 0
\]
Therefore, for sufficiently large \(n\),
\[
    \widehat\mu_n^{\rho}
    \bigl(F_\varepsilon(\rho)\bigr) \le \exp(-\frac{n}{2}\inf_{\sigma\in F_\varepsilon(\rho)}
    I(\sigma\|\rho))
    \longrightarrow 0,
\]
which is exactly consistency in Definition~\ref{def:consistency}.
\end{proof}

In fact, the consistency condition defined by Definition~\ref{def:consistency} is strictly weaker than the unique-zero property~\eqref{eq:unique-zero}. To see this, suppose $E_n$ is a protocol with rate function satisfying the unique-zero property. Then for a fixed $\tau \in \mc D_d$, define
\[
\widetilde E_n(A)
    :=
    \left(1-\frac1n\right)E_n(A)
    +
    \frac1n\,\mathbf 1_{\{\tau\in A\}}\,
    \mathbf 1_{\mathcal H_n},
    \qquad
    A\in\mathfrak B(\mathcal D_d).
\]
One can verify that $\widetilde E_n$ is consistent, but the rate function always admits $\tau$ as a zero, breaking the unique-zero property. Indeed, the unique-zero condition rules out
subexponentially small probabilities of reporting wrong states, while ordinary consistency does not; nevertheless, our main theorem shows that Keyl's covariant exponent remains optimal under this weaker consistency requirement.
\medskip
\begin{definition}[Tomographic rate functions]
A function
\[
    I:\mathcal D_d\times\mathcal D_d\to[0,\infty]
\]
is called a tomographic rate function if there exists a consistent tomography
protocol such that, for every true state
\(\rho\in\mathcal D_d\), the estimator laws
\((\widehat\mu_n^{\rho})_{n\ge1}\) satisfy a LDP with rate function
\(I(\cdot\|\rho)\).
\end{definition}
We denote 
\begin{itemize}
    \item \(\mathcal E\) as the set of all tomographic rate functions. 
    \item \(\widehat{\mathcal E}\) as the set of all tomographic rate functions satisfying unique-zero property \eqref{eq:unique-zero}. 
    \item \(\mathcal E^0\subseteq\mathcal E\) as the subset consisting of those \(I\in\mathcal E\) such that, for every \(\sigma\in\mathcal D_d\), the map $\rho\mapsto I(\sigma\|\rho)$ is lower semicontinuous.
    \item \(\mathcal E^c\subseteq\mathcal E\) as the subset consisting of those rate functions arising from covariant protocols.
\end{itemize}
\medskip
\begin{conjecture}[Keyl's conjecture]\label{conj:Keyl}
    Fix any $\sigma,\rho \in \mc D_d$, define
    \begin{equation}
     \mathcal I_{\mathrm{Id}}^c(\sigma \| \rho) = \sup_{I \in \mc E^c} I(\sigma\| \rho),\quad \mathcal I_{\mathrm{Id}}^0(\sigma \| \rho) = \sup_{I \in \mc E^0} I(\sigma\| \rho),\quad \mathcal I_{\mathrm{Id}}(\sigma \| \rho) = \sup_{I \in \widehat{\mc E}} I(\sigma\| \rho).
\end{equation}
Then
\begin{equation}\label{eq:Keyl-conjecture}
    S^{\mathrm{Keyl}}(\sigma \| \rho)
    =
    \mathcal I_{\mathrm{Id}}^c(\sigma \| \rho)
    =
    \mathcal I_{\mathrm{Id}}^0(\sigma \| \rho)
    \le
    \mathcal I_{\mathrm{Id}}(\sigma \| \rho)
    =
    D(\sigma \| \rho),
\end{equation}

\noindent where $D(\sigma \| \rho)$ is the relative entropy, and $S^{\mathrm{Keyl}}(\sigma \| \rho)$ is defined in~\eqref{eq:Keyl-rate-intro}.
\end{conjecture}
We explain more details about $S^{\mathrm{Keyl}}(\sigma \| \rho)$. For any $x \in \Delta_d^\downarrow,\ \tau \in \mc D_d$, let $\Delta_0(\tau):=1$ and let $\Delta_k(\tau)$ be the determinant of the \(k\times k\) principal minor of $\tau$, defined in \eqref{eq:determinant-principle-minor}. Set
\begin{equation}\label{eq:principle-ratio-vector}
    \Theta_k(\tau):= \lim_{\varepsilon \to 0+}\frac{\Delta_k(\tau + \varepsilon \textbf{1})}{\Delta_{k-1}(\tau+ \varepsilon \textbf{1})},\quad 1\le k \le d,
\end{equation}
We introduce the following function
\begin{equation}\label{eq:Keyl-rate}
    L_{\mathrm{Keyl}}(x,\tau) = 
         -\sum_{k=1}^d x_k \log \Theta_k(\tau),
\end{equation}
with the convention $0 \log 0 = 0$ and $-u \log 0 = \infty$ for $u > 0$.

Then by Abel's summation formula, we also have $$S^{\mathrm{Keyl}}(\sigma \| \rho) = -H(x) + L_{\mathrm{Keyl}}(x,U^\dagger \rho U),\quad \sigma = U\diag(x)U^\dagger. $$ 
Note that the definition of $S^{\mathrm{Keyl}}(\sigma \| \rho)$ does not depend on the choice of $U$ with $\sigma = U\diag(x)U^\dagger, x \in \Delta_d^\downarrow$, see~\cite[Lemma 4.15]{Keyl_2006}. 

For commuting states, $S^{\mathrm{Keyl}}(\sigma \| \rho) = D(\sigma \| \rho)$. Furthermore, it was shown that $S^{\mathrm{Keyl}}(\sigma \| \rho) < D(\sigma \| \rho)$ for some noncommuting states $\sigma,\rho$, see the discussions after Theorem 3.3 in \cite{Keyl_2006}. Here we show a strengthened result:

\begin{proposition}[Comparison with relative entropy]\label{prop:Keyl-vs-relative-entropy}
For all \(\sigma,\rho\in\mathcal D_d\),
\[
    S^{\mathrm{Keyl}}(\sigma\|\rho) \le D(\sigma\|\rho)
\]
as an inequality in \([0,\infty]\). Moreover,
\begin{enumerate}
    \item If \(D(\sigma\|\rho)<\infty\), then
    \[
        S^{\mathrm{Keyl}}(\sigma\|\rho)
        =
        D(\sigma\|\rho) \iff [\sigma,\rho]=0.
    \]
    \item If \(D(\sigma\|\rho)=+\infty\), then
    \[
        S^{\mathrm{Keyl}}(\sigma\|\rho)=D(\sigma\|\rho) = +\infty \iff \supp(\sigma)\cap\ker(\rho)\neq\{0\}.
    \]
\end{enumerate}
\end{proposition}
A proof of this result, together with some additional properties of $S^{\mathrm{Keyl}}$ are presented in Appendix \ref{section:Appendix_KeylsRateProperties}.

\section{Non-covariant protocol achieving relative entropy as rate function}\label{sec:relative-entropy}
In this section, we prove Theorem~\ref{thm:relative-entropy-envelope-intro} thus the pointwise relative entropy envelope predicted in the last equality in Conjecture~\ref{conj:Keyl}. 

Our contribution is to show that for every pair
\((\sigma,\rho)\) and every \(a<D(\sigma\|\rho)\), one can construct a non-covariant tomography protocol whose rate at \((\sigma,\rho)\) is larger than \(a\) and has unique-zero property. 

Recall that every tomographic LDP rate with unique-zero property~\eqref{eq:unique-zero} is bounded above by the quantum relative entropy, which is proved in \cite[Theorem 3.3]{Keyl_2006}. We present a proof under the consistency condition only, slightly weakening the assumption:
\medskip
\begin{proposition}[Converse bound]
Let \(I\in\mathcal E\) be the rate function of any consistent tomography
protocol.  Then for all \(\sigma,\rho\in\mathcal D_d\),
\[
    I(\sigma\|\rho)\le D(\sigma\|\rho).
\]
\end{proposition}

\begin{proof}
If \(D(\sigma\|\rho)=+\infty\), the claim is trivial.  Assume
\(D(\sigma\|\rho)<+\infty\).  For \(\varepsilon>0\), set
\[
    F_\varepsilon:=\{\tau:\|\tau-\sigma\|_1\le \varepsilon\}.
\]
By consistency under the true state \(\sigma\),
\[
    \Tr(E_n(F_\varepsilon)\sigma^{\otimes n})\longrightarrow 1.
\]
Thus \(E_n(F_\varepsilon)\) is a sequence of tests whose type-I success
probability under \(\sigma^{\otimes n}\) tends to one.  By the converse part of quantum Stein's lemma \cite{hiai1991proper, ogawa2000strong},
\[
    \limsup_{n\to\infty}
    -\frac1n\log \Tr(E_n(F_\varepsilon)\rho^{\otimes n})
    \le
    D(\sigma\|\rho).
\]
On the other hand, applying the closed-set LDP bound to \(F_\varepsilon\)
under the true state \(\rho\) gives
\[
    \liminf_{n\to\infty}
    -\frac1n\log \Tr(E_n(F_\varepsilon)\rho^{\otimes n})
    \ge
    \inf_{\tau\in F_\varepsilon} I(\tau\|\rho).
\]
Therefore $\inf_{\tau\in F_\varepsilon} I(\tau\|\rho)
    \le
    D(\sigma\|\rho)$. Since \(I(\cdot\|\rho)\) is lower semicontinuous,
\[
    I(\sigma\|\rho)
    =
    \lim_{\varepsilon\downarrow0}
    \inf_{\tau\in F_\varepsilon} I(\tau\|\rho) \le D(\sigma\|\rho).
\]
\end{proof}
We conclude the lower bound by proving a more general result, which embeds any Stein exponent into a rate function of a tomography protocol:
\medskip
\begin{theorem}[Embedding Stein exponents into tomography]
\label{thm:testing-to-tomography}
Let \(\sigma_0,\rho_0\in\mathcal D_d\).  Suppose that
\(s\in[0,\infty]\) is achievable by a sequence of binary tests $\{M_n\}_{n\ge 1}$:
\begin{equation}\label{eq:achievable-exponent}
        \Tr(M_n\sigma_0^{\otimes n})\longrightarrow 1,\quad \liminf_{n\to\infty}
    -\frac1n
    \log\Tr(M_n\rho_0^{\otimes n})
    \ge s.
\end{equation}
Then, for every finite \(a<s\), there exists a consistent tomography protocol satisfying an LDP with rate function $I_a$ such that $I_a(\sigma\|\rho)=0 \iff
    \sigma = \rho$ and
\[
    I_a(\sigma_0\|\rho_0)>a.
\]
\end{theorem}
\begin{proof}
Fix \(a<s\). Choose \(\gamma\in(0,1)\) small enough such that
\(a/(1-\gamma)<s\). From the assumption~\eqref{eq:achievable-exponent}, we can choose a block length
\(\ell\) and an effect \(M: = M_{\ell}\) such that, with
\[
    \alpha=\Tr(M\sigma_0^{\otimes\ell}),
    \qquad
    \beta=\Tr(M\rho_0^{\otimes\ell}),
\]
we have
\[
    \frac{1-\gamma}{\ell}d(\alpha\|\beta)>a, 
\]
where the binary relative entropy is defined by 
\[
d(\alpha\|\beta) = D((\alpha,1-\alpha) \| (\beta, 1-\beta)) = -\alpha \log \beta - (1-\alpha) \log (1-\beta) - h_2(\alpha). 
\]
To see this, note that $\alpha$ is close to $1$ and $-\log \beta$ is greater than $a \ell$ for large $\ell$. 

By continuity of the binary relative entropy, we can choose \(\eta\) small enough so that
\[
    \frac{1-\gamma}{\ell} \inf_{|q-\alpha|\le\eta}d(q\|\beta)>a .
\]
Now split \(n\) copies into
\[
    m_n=\lfloor\gamma n\rfloor,
    \qquad
    k_n=\left\lfloor\frac{n-m_n}{\ell}\right\rfloor .
\]
On the first \(m_n\) copies, run any baseline consistent tomography whose estimator
\(\widehat \rho_n\) satisfies an LDP with rate function \(K\) and
\[
    K(\sigma\|\rho)=0 \iff \sigma=\rho .
\]
For example, one may choose Keyl's protocol. On the remaining \(k_n\) blocks, measure \(\{M,\mathbf1-M\}\), obtain
\(Z_{n,1},\ldots,Z_{n,k_n}\), and set
\[
    \bar Z_n=k_n^{-1}\sum_i Z_{n,i}.
\]
Then
\[
    Y_n=(\widehat \rho_n,\bar Z_n)
\]
satisfies an LDP with rate
\begin{equation}\label{eq:rate-joint-distribution}
        R(\sigma,p \| \rho)=\gamma K(\sigma\|\rho)+\frac{1-\gamma}{\ell}d(p\|p_\rho).
\end{equation}
Recall for any state $\tau$, $p_\tau: = \Tr(M \tau^{\otimes \ell})$. Define a post-processing map $\Phi: \mc D_d \times [0,1] \to \mc D_d$
\[
\Phi_{\omega_*}(\tau,p)=
    \begin{cases}
    \tau, & |p - p_\tau|\le\eta,\\
    \omega_*, & |p - p_\tau|>\eta,
     \end{cases}
\]
where we choose \(\omega_*\neq\sigma_0\). The final estimator is given by 
\begin{equation}\label{eq:estimator-relative-entropy}
    \widehat\theta_n= \Phi(\widehat \rho_n,\bar{Z}_n).
\end{equation}
Using the general contraction principle in Theorem~\ref{thm:contraction-LDP}, and the details are covered in Appendix~\ref{app:LDP}, we can show that $\widehat\theta_n$ satisfies a LDP with rate function 
\[
I_a(\sigma\| \rho) = \inf\left\{ R(\sigma,q \|\rho):\ q\in[0,1], |q-p_\sigma| \le \eta \right\}, \quad \sigma\neq\omega_*,
\]
where $R(\sigma,q \|\rho)$ is defined in \eqref{eq:rate-joint-distribution}. At $\sigma = \omega_*$, we have 
\[
I_a(\sigma \|\rho)= \min\bigg\{ \inf\left\{ R(\sigma,q \|\rho):\ q\in[0,1], |q-p_{\sigma}|\le\eta\right\},\ \inf\left\{ R(\tau,q \|\rho):\ q\in[0,1],\ \tau \in \mc D_d,\ |q-p_{\tau}|\ge\eta\right\} \bigg\}.
\]
Note that \(I_a\) has the unique-zero property. Away from \(\omega_*\),
the term \(K(\sigma\|\rho)\) forces \(\sigma = \rho\). At \(\sigma = \omega_*\), because \(R_\rho\) has unique zero \((\rho,p_\rho)\), $I_a(\omega_* \| \rho) = 0$ if and only if $\omega_* = \rho$.

At the target pair $(\sigma_0,\rho_0)$, with $\sigma_0 \neq \omega_*$, we have
\[
\begin{aligned}
    I_a(\sigma_0\|\rho_0)
    &\ge
    \frac{1-\gamma}{\ell}
    \inf_{|p-\alpha|\le\eta}d(p\|\beta)
    >a,
\end{aligned}
\]
which concludes the proof.
\end{proof}

\section{Optimality of Keyl's rate function for covariant protocols}\label{sec:main}
In this section, we prove Theorem~\ref{thm:main-introduction} thus the optimality of Keyl's rate for covariant protocols predicted in the first equality in Conjecture~\ref{conj:Keyl}. 


Namely, we show that if a covariant protocol is consistent and its estimator distributions satisfy an LDP, then its rate function cannot exceed Keyl's rate at any point. Thus Keyl's rate is a universal pointwise upper bound on the large-deviation rate function achievable by covariant tomography protocols.
\medskip
\begin{theorem}[Converse bound for covariant tomography]\label{thm:main-covariant}
Let \(\{E_n\}_{n\ge 1}\) be a covariant quantum state tomography
protocol. Assume that it is consistent in the sense of Definition~\ref{def:consistency}, and for each \(\rho\in\mathcal D_d\), the estimator
distributions \(\{\widehat\mu_n^\rho\}_{n\ge1}\) satisfy a large deviation
principle on \(\mathcal D_d\) with rate function
\(I(\cdot\|\rho)\). Then, for every
\(\rho,\sigma\in\mathcal D_d\),
\begin{equation}\label{eq:main-covariant-bound}
    I(\sigma\|\rho)
    \le
    S^{\mathrm{Keyl}}(\sigma\|\rho).
\end{equation}
\end{theorem}
To prove the theorem, suppose we are given an arbitrary covariant quantum state tomography protocol. Following the general characterization in Theorem~\ref{thm:characterization-covariant}, we can assume that the protocol is given by $\left(\{M_\sigma^{(n)}\}_{\sigma \in \mc D_d},\mu_n \right)_{n\ge 1}$ where $\mu_n$ is unitarily invariant and $\{M_\sigma^{(n)}\}_{\sigma \in \mc D_d}$ has the form
    \begin{equation}
        M_{U\diag(x)U^\dagger}^{(n)}
    =
    U_{\mathrm{Schur}}^\dagger
    \left(
        \sum_{\lambda\vdash_d n} \dim \mc Q_\lambda
        |\lambda\rangle\langle\lambda|
        \otimes \mathbf 1_{\mc P_\lambda}
        \otimes
        q_\lambda(U)K_{\lambda,x}q_\lambda(U^\dagger)
    \right)
    U_{\mathrm{Schur}}.
    \end{equation}
with $\int_{\Delta_d^\downarrow}
    \Tr K_{\lambda,x}d\nu_n(x) = 1,\ \forall \lambda \vdash_d n$, where $\nu_n$ is defined by \eqref{eq:mu-disintegration}. The corresponding estimator distribution is denoted by
\(\widehat\mu_n^{\rho,K}\):
\begin{equation}\label{eq:probability-estimator-covariant}
    \begin{aligned}
    \widehat\mu^{\rho,K}_n(A)& := \int_A \Tr(M_\sigma^{(n)} \rho^{\otimes n}) d\mu_n(\sigma) \\
    & = \sum_{\lambda \vdash_d n} \dim \mc P_\lambda \dim \mc Q_\lambda\Tr(q_\lambda(\rho)\iint_{U\diag(x)U^\dagger \in A} q_\lambda(U) K_{\lambda, x} q_\lambda(U^\dagger) d\nu_n(x) d\mu_{\mathrm{Haar}}(U)).
\end{aligned}
\end{equation}
A local analysis of the probability of a neighborhood of an estimator is crucial. To this end, we introduce the localized seed operator 
\begin{equation}\label{eq:averaged-seed}
    K_{\lambda,x}^\varepsilon:= \iint_{W \diag(y) W^\dagger \in N_\varepsilon(\diag(x))} q_\lambda(W) K_{\lambda, y} q_\lambda(W^\dagger) d\nu_n(y) d\mu_{\mathrm{Haar}}(W).
\end{equation}
One can check that \(K_{\lambda,x}^\varepsilon\) is positive,
torus-invariant, and satisfies
\[
    \Tr K_{\lambda,x}^\varepsilon\le 1.
\]

The following implication of consistency is crucial:
\medskip
\begin{lemma}\label{lem:consistent-implication}
    If the protocol is consistent, i.e., for any $\rho$ and $\varepsilon > 0$, $\lim_{n\to \infty}\widehat\mu^{\rho,K}_n(N_\varepsilon(\rho)) =1$. Then for any $x\in \Delta_d^\downarrow$ and $\varepsilon > 0$, there exists a Young sequence $\lambda^{(n)}$ with $\lambda^{(n)}/n \to x$ and $\ell(\lambda^{(n)}) \le \rank(x)$, such that
    \begin{equation}
        \lim_{n\to \infty} -\frac{1}{n} \log \Tr(q_{\lambda^{(n)}}(\diag(x)) K_{\lambda^{(n)},x}^\varepsilon) = H(x).
    \end{equation}
\end{lemma}
\begin{proof}
We extend the probability measures $\widehat\mu^{\rho,K}_n$ on $\mc D_d$ to probability measures on $\Delta_d^\downarrow \times \mc D_d$. For any $\lambda \vdash_d n$ and $A \in \mathfrak B(\mc D_d)$, we define
\[
    \widehat\mu^{\rho,K}_n(\lambda,A):= \dim \mc P_\lambda \dim \mc Q_\lambda\Tr(q_\lambda(\rho)\iint_{U\diag(x)U^\dagger \in A} q_\lambda(U) K_{\lambda, x} q_\lambda(U^\dagger) d\nu_n(x) d\mu_{\mathrm{Haar}}(U)).
\]
Then by definition \eqref{eq:probability-estimator-covariant}, we have $\widehat\mu_n^{\rho,K}(A)=\sum_{\lambda\vdash_d n} \widehat\mu^{\rho,K}_n(\lambda,A)$. 

Note that the marginal distribution of \(\lambda\) is given by $\mb P_{n,\rho}$, defined in \eqref{eq:probability-Young}. Indeed, taking \(A=\mc D_d\), we get
\[
\begin{aligned}
    \widehat\mu_n^{\rho,K}(\lambda,\mc D_d)
    &= \dim \mc P_\lambda \dim \mc Q_\lambda
    \Tr\left(
        q_\lambda(\rho)
        \int_{\Delta_d^\downarrow}\int_{\U(d)}
        q_\lambda(U)K_{\lambda,y} q_\lambda(U^\dagger)
        d\mu_{\mathrm{Haar}}(U)d\nu_n(y)
    \right) \\
    & = \dim \mc P_\lambda \Tr(q_\lambda(\rho)) = \mb P_{n,\rho}(\lambda),
\end{aligned}
\]
where we used Lemma~\ref{lemma:Schur-orthogonality} and the normalization $\int_{\Delta_d^\downarrow} \Tr(K_{\lambda,y})d\nu_n(y) = 1$. In particular, 
\begin{equation}\label{eq:comparison-Young-distribution}
    \widehat\mu^{\rho,K}_n(\lambda,A) \le P_{n,\rho}(\lambda),\quad A \in \mathfrak B(\mc D_d).
\end{equation}
For any $\delta > 0$, recall that
\begin{equation}
    G_{n,\delta}(x):= \{\lambda: \|\lambda/n  -x\|_1 < \delta, \ \ell(\lambda) \le \rank(x)\}.
\end{equation}
Via Lemma~\ref{lemma:Schur-concentration}, we have
\[
\lim_{n\to \infty}\mb P_{n,\rho}(G^c_{n,\delta}(x)) = 0.
\]
Therefore, for $A = N_\varepsilon(\diag(x))$, we have 
\[
\begin{aligned}
     \widehat\mu^{\diag(x),K}_n(N_\varepsilon(\diag(x))) & = \sum_{\lambda\vdash_d n} \widehat\mu^{\diag(x),K}_n(\lambda, N_\varepsilon(\diag(x))) \\
     & = \sum_{\lambda \in G_{n,\delta}(x)} \widehat\mu^{\diag(x),K}_n(\lambda, N_\varepsilon(\diag(x))) + \sum_{\lambda \in G^c_{n,\delta}(x)} \widehat\mu^{\diag(x),K}_n(\lambda, N_\varepsilon(\diag(x))) \\
     & \le \sum_{\lambda \in G_{n,\delta}(x)} \widehat\mu^{\diag(x),K}_n(\lambda, N_\varepsilon(\diag(x))) + \mb P_{n,\rho}(G^c_{n,\delta}(x))
\end{aligned}
\]
where the last inequality follows from \eqref{eq:comparison-Young-distribution}. Recall that $\widehat\mu^{\diag(x),K}_n(N_\varepsilon(\diag(x))) \to 1$, then for each $n > 0$, we choose $\delta_n \to 0$ such that
\[
\lim_{n \to \infty}\sum_{\lambda \in G_{n,\delta_n}(x)} \widehat\mu^{\diag(x),K}_n(\lambda, N_\varepsilon(\diag(x))) = 1. 
\]
Since \[
\widehat\mu^{\diag(x),K}_n(\lambda, N_\varepsilon(\diag(x))) =  \dim \mc P_\lambda \dim \mc Q_\lambda \Tr(q_\lambda(\diag(x)) K_{\lambda,x}^\varepsilon),
\]
we have 
\[
\sum_{\lambda \in G_{n,\delta_n}(x)} \dim \mc P_\lambda \dim \mc Q_\lambda \Tr(q_\lambda(\diag(x)) K_{\lambda,x}^\varepsilon) \ge 1 - o(1).
\]
Via Lemma~\ref{lemma:rep-data}, for each $n\ge 1$, there exists $\lambda^{(n)} \in G_{n,\delta_n}(x)$, such that 
\[
\Tr(q_{\lambda^{(n)}}(\diag(x)) K_{{\lambda^{(n)}},x}^\varepsilon) \ge e^{-n H(\overline{{\lambda^{(n)}}})} \mathrm{poly}(n)^{-1}.
\]
On the other hand, we always have \[
\dim \mc P_{\lambda^{(n)}} \dim \mc Q_{\lambda^{(n)}} \Tr(q_{\lambda^{(n)}}(\diag(x)) K_{{\lambda^{(n)}},x}^\varepsilon) \le 1 \implies \Tr(q_{\lambda^{(n)}}(\diag(x)) K_{{\lambda^{(n)}},x}^\varepsilon) \le e^{-n H(\overline{{\lambda^{(n)}}})} \mathrm{poly}(n).
\]
Taking $\lim_{n \to \infty} -\frac{1}{n} \log$ on both sides concludes the proof. 
\end{proof}
The following lower bound is crucial for the weight analysis:
\medskip
\begin{lemma}[Principal minor lower bound] \label{lem:priciple-minor-lower-bound}
Let \(\tau \ge 0\), then for every \(\alpha\in\mathrm{Wt}_d(\lambda)\),
\[
    P_{\lambda,\alpha}q_\lambda(\tau)P_{\lambda,\alpha}
    \ge
    \Theta(\tau)^\alpha P_{\lambda,\alpha},
    \qquad
    \Theta(\tau)^\alpha:=\prod_{k=1}^d \Theta_k(\tau)^{\alpha_k}.
\]
\end{lemma}
\begin{proof}
Assume $\tau>0$ first. Using the Cholesky decomposition, there exists an upper triangular matrix $L$ with diagonal entries equal to \(1\), such that 
\[
    \tau=L^\dagger \diag(\Theta(\tau)) L.
\]
Let \(v\in\mathcal Q_{\lambda,\alpha}\). By the triangular weight property~\eqref{eq:triangular-weight},
\[
    q_\lambda(L)v
    =
    v+\sum_{\beta\succ\alpha} v_\beta,
    \qquad
    v_\beta\in\mathcal Q_{\lambda,\beta}.
\]
Since \(q_\lambda(\diag(\Theta(\tau)))\) is diagonal in the weight-space decomposition and
different weight spaces are orthogonal,
\[
\begin{aligned}
    \langle v,q_\lambda(\tau)v\rangle
    &=
    \left\langle
        q_\lambda(L)v,
        q_\lambda(\diag(\Theta(\tau)))q_\lambda(L)v
    \right\rangle  \\
    &=
    \Theta(\tau)^\alpha \|v\|^2
    +
    \sum_{\beta\succ\alpha}
    \Theta(\tau)^\beta \|v_\beta\|^2 \\
    &\ge
    \Theta(\tau)^\alpha\|v\|^2.
\end{aligned}
\]
This proves the operator inequality. When $\tau \ge 0$, replacing $\tau$ by $\tau + \varepsilon \textbf{1}$ and then sending $\varepsilon \to 0+$ concludes the proof.
\end{proof}
The following universal upper bound is derived:
\medskip
\begin{lemma}[Comparison lemma]
\label{lem:robust-Keyl-comparison-torus-invariant}
Let \(x\in\Delta_d^\downarrow\). Suppose \(\lambda^{(n)}\vdash_d n\)
satisfies
\[
    \frac{\lambda^{(n)}}{n}\longrightarrow x,\quad \ell(\lambda^{(n)}) \le \rank(x)
\]
and suppose \(K_{\lambda^{(n)}} \in \mc L(\mc Q_{\lambda^{(n)}})_+\) satisfying $\Tr K_{\lambda^{(n)}}\le 1$ and 
\[
   q_{\lambda^{(n)}}(T) K_{\lambda^{(n)}}q_{\lambda^{(n)}}(T^\dagger) =K_{\lambda^{(n)}},\ \forall T \in \mathrm{T}(d).  
\]
Assume that
\begin{equation}\label{eq:robust-self-consistency}
    \limsup_{n\to\infty}
    -\frac1n
    \log
    \Tr\left(
        q_{\lambda^{(n)}}(\diag(x))K_{\lambda^{(n)}}
    \right)
    \le
    H(x).
\end{equation}
Then, for every \(\tau\in\mc D_d\),
\begin{equation}\label{eq:robust-Keyl-comparison-conclusion}
    \limsup_{n\to\infty}
    -\frac1n
    \log
    \Tr\left(
        q_{\lambda^{(n)}}(\tau)
        K_{\lambda^{(n)}}
    \right)
    \le
    L_{\mathrm{Keyl}}(x,\tau),
\end{equation}
where $L_{\mathrm{Keyl}}(x,\tau)$ is defined in~\eqref{eq:Keyl-rate}.
\end{lemma}

\begin{proof}
For brevity, we write $\lambda = \lambda^{(n)}$ and 
\[
r := \rank(x) = \#\{i:x_i>0\}.
\]
We assume $\rank(\tau)\ge r$, otherwise \(L_{\mathrm{Keyl}}(x,\tau)=+\infty\) and \eqref{eq:robust-Keyl-comparison-conclusion} clearly holds. For any fixed $\delta>0$, we introduce 
\begin{equation}\label{eq:quantity-omega}
    \omega_{x,\tau}(\delta):= \sup\left\{\left|\sum_{k=1}^r (y_k - x_k)\log \Theta_k(\tau) \right|:\ y\in \Delta_d,\ \rank(y) \le r,\ \|y-x\|_1 < \delta \right\}.
\end{equation}
If \(\Theta_k(\tau)=0\) for some \(k\le r\), then
\(L_{\mathrm{Keyl}}(x,\tau)=+\infty\), and there is nothing to prove. Hence we
may assume that
\[
    \Theta_k(\tau)>0,
    \qquad k=1,\ldots,r.
\]
Under this assumption the quantities \(\log\Theta_k(\tau)\) are finite, and
\[
    \lim_{\delta\to0+}\omega_{x,\tau}(\delta)=0.
\]
Denote 
\begin{equation}\label{eq:weight-close-to-x}
    G_{n,\delta}^\lambda(x):= \left\{\alpha \prec \lambda: \left\|\frac{\alpha}{n} - x\right\|_1 < \delta\right\}.
\end{equation}
We claim
\begin{equation}\label{eq:Keyl-comparison-weight-claim}
    \limsup_{n\to \infty} -\frac{1}{n} \log \bigg( \sum_{\alpha \in G_{n,\delta}^\lambda(x)} \Tr(P_{\lambda,\alpha} K_\lambda)\bigg) \le \omega_{x,\diag(x)}(\delta).
\end{equation}
In fact, for any $\alpha \in G_{n,\delta}^\lambda(x)$, $\alpha/n$ satisfies the constraint in \eqref{eq:quantity-omega}, then by definition, we have
\begin{equation}\label{eq:quantity-omega-implication}
\left|\sum_{k=1}^r \left(\frac{\alpha_k}{n} - x_k\right) \log \Theta_k(\tau) \right| \le \omega_{x,\tau}(\delta).
\end{equation}
Choosing $\tau = \diag(x)$ in \eqref{eq:quantity-omega-implication}, we have $\Theta_k(\tau) = x_k$ thus taking the exponential we have
\begin{equation}\label{eq:monomial-upper-bound-high}
    x^\alpha \le \exp(-nH(x) + n\, \omega_{x,\diag(x)}(\delta)),\quad \alpha \in G_{n,\delta}^\lambda(x).
\end{equation}
If $\alpha \notin G_{n,\delta}^\lambda(x)$, denote $\overline{\alpha} = \alpha /n$, we have $\|\alpha/n - x\|_1 \ge \delta$. Note that such $\alpha$ may not exist and in this case we set any quantity involving $\sum_{\alpha \notin G_{n,\delta}^\lambda(x)}$ to be zero. Without loss of generality, we assume for all $n$, such $\alpha$ exists. Then via Schur concavity of von Neumann entropy, we have $H(\overline{\alpha}) \ge H(\overline{\lambda})$, and
\[
\begin{aligned}
    \sum_{k=1}^r \left(\frac{\alpha_k}{n} - x_k\right) \log x_k & = H(x) - H(\overline{\alpha})-D(\overline{\alpha}\| x)\\ 
    & \le H(x) - H(\overline{\lambda})-D(\overline{\alpha}\| x).
\end{aligned}
\]
Since $\overline \lambda \to x$, as $n$ large enough, we can make $H(x) - H(\overline{\lambda})$ arbitrarily close to zero. On the other hand, via Pinsker's inequality, $D(\overline{\alpha}\| x)$ is bounded away from zero. Thus whenever $n$ is large enough, there exists $c_\delta > 0$ such that 
\[
\sum_{k=1}^r \left(\frac{\alpha_k}{n} - x_k\right) \log x_k  \le -c_\delta,
\]
which implies 
\begin{equation}\label{eq:monomial-upper-bound-low}
    x^\alpha = \exp(-nH(x) + n\sum_{k=1}^r \left(\frac{\alpha_k}{n} - x_k\right) \log x_k) \le \exp(-nH(x) - n\, c_\delta),\quad \alpha \notin G_{n,\delta}^\lambda(x).
\end{equation}
Then we have
\[
\begin{aligned}
    \Tr(q_\lambda(\diag(x)) K_\lambda)& = \sum_{\alpha \prec \lambda} x^\alpha \Tr(P_{\lambda,\alpha} K_\lambda) \\
    & =  \sum_{\alpha \in G_{n,\delta}^\lambda(x)} x^\alpha \Tr(P_{\lambda,\alpha} K_\lambda) + \sum_{\alpha \notin G_{n,\delta}^\lambda(x)} x^\alpha \Tr(P_{\lambda,\alpha} K_\lambda) \\
    & \le \exp(-nH(x) + n\, \omega_{x,\diag(x)}(\delta)) \sum_{\alpha \in G_{n,\delta}^\lambda(x)} \Tr(P_{\lambda,\alpha} K_\lambda) + \exp(-nH(x) - n\, c_\delta).
\end{aligned}
\]
For the first equality, we used \eqref{eq:diagonal-decomposition}. The second equality follows from separating the weights into $G_{n,\delta}^\lambda(x)$ and $G_{n,\delta}^\lambda(x)^c$. For the inequality, we used \eqref{eq:monomial-upper-bound-high} and \eqref{eq:monomial-upper-bound-low} respectively, and since $\Tr(K_\lambda) \le 1$, we have $\sum_{\alpha \notin G_{n,\delta}^\lambda(x)} \Tr(K_\lambda P_{\lambda, \alpha}) \le \sum_{\alpha \prec \lambda} \Tr(K_\lambda P_{\lambda, \alpha}) \le 1$. On the other hand, since 
\[
\limsup_{n\to\infty}
    -\frac1n\log\Tr\left(
        q_{\lambda}(\diag(x))K_{\lambda}\right) \le H(x),
\]
there exists $\eta_n \to 0$, such that 
\[
\Tr(q_\lambda(\diag(x)) K_\lambda) \ge \exp(-nH(x) - n\, \eta_n).
\]
Combining the two bounds, we get 
\[
\sum_{\alpha \in G_{n,\delta}^\lambda(x)} \Tr(P_{\lambda,\alpha} K_\lambda) \ge (e^{-n\, \eta_n} - e^{-n\,c_\delta})e^{-n\,\omega_{x,\diag(x)}(\delta)}.
\]
Taking $-\frac{1}{n} \log$ on both sides concludes the proof of \eqref{eq:Keyl-comparison-weight-claim}.

Now we use the fact that $K_\lambda$ is torus-invariant, thus by \eqref{eq:twirl-torus}, $K_\lambda = \sum_{\alpha \prec \lambda} P_{\lambda,\alpha} K_\lambda P_{\lambda,\alpha}$. Then 
\[
\begin{aligned}
    \Tr(q_\lambda(\tau) K_\lambda)& = \sum_{\alpha \in G_{n,\delta}^\lambda(x)} \Tr(P_{\lambda,\alpha} q_\lambda(\tau) P_{\lambda,\alpha} K_\lambda) + \sum_{\alpha \notin G_{n,\delta}^\lambda(x)} \Tr(P_{\lambda,\alpha} q_\lambda(\tau) P_{\lambda,\alpha} K_\lambda) \\
    & \ge \sum_{\alpha \in G_{n,\delta}^\lambda(x)} \Tr(P_{\lambda,\alpha} q_\lambda(\tau) P_{\lambda,\alpha} K_\lambda) \\
    & \ge \sum_{\alpha \in G_{n,\delta}^\lambda(x)} \Theta(\tau)^\alpha\Tr(P_{\lambda,\alpha} K_\lambda),
\end{aligned}
\]
where the last inequality follows from $P_{\lambda,\alpha} q_\lambda(\tau) P_{\lambda,\alpha} \ge \Theta(\tau)^\alpha P_{\lambda,\alpha}$ shown in Lemma~\ref{lem:priciple-minor-lower-bound}. Recall the definition $\omega_{x,\tau}(\delta)$ in \eqref{eq:quantity-omega}, and $\alpha \in G_{n,x}^\lambda(\delta)$ satisfies $\|\alpha /n - x\|_1 < \delta$, thus 
\[
\begin{aligned}
    \Theta(\tau)^\alpha & = \exp(n\sum_{k=1}^r \frac{\alpha_k}{n} \log \Theta_k(\tau)) \\
    & = \exp(n\sum_{k=1}^r x_k \log \Theta_k(\tau) + n\sum_{k=1}^r (\frac{\alpha_k}{n} - x_k) \log \Theta_k(\tau) ) \\
    & \ge \exp(n\sum_{k=1}^r x_k \log \Theta_k(\tau) - n\,\omega_{x,\tau}(\delta)).
\end{aligned}
\]
Therefore, 
\[
\begin{aligned}
   & \limsup_{n\to \infty} -\frac{1}{n} \log \Tr(q_\lambda(\tau) K_\lambda) \\
   & \le -\sum_{k=1}^r x_k \log \Theta_k(\tau) + \omega_{x,\tau}(\delta) + \limsup_{n\to \infty} -\frac{1}{n} \log \bigg( \sum_{\alpha \in G_{n,\delta}^\lambda(x)} \Tr(P_{\lambda,\alpha} K_\lambda)\bigg) \\
   & \le -\sum_{k=1}^r x_k \log \Theta_k(\tau) + \omega_{x,\tau}(\delta) + \omega_{x,\diag(x)}(\delta). 
\end{aligned}
\]
Finally, sending $\delta \to 0$, we conclude the proof.
\end{proof}
To transfer the derived local bound to pointwise rate function bounds, we need the following lemma:
\medskip
\begin{lemma}[From local bounds to pointwise rate bounds]
\label{lem:local-bound-to-rate-bound}
Let \(\{\widehat\mu_n^{\rho,K}\}_{n\ge1}\) satisfy an LDP with rate
function \(I_K(\cdot\|\rho)\). Fix \(\sigma\in\mc D_d\). Suppose that
for some \(R\in[0,\infty]\), for every \(\varepsilon>0\),
\begin{equation}\label{eq:local-ball-exponential-bound}
    \limsup_{n\to\infty}
    -\frac1n\log
    \widehat\mu_n^{\rho,K}(N_\varepsilon(\sigma))
    \le R .
\end{equation}
Then
\begin{equation}\label{eq:local-bound-rate-conclusion}
    I_K(\sigma\|\rho)\le R .
\end{equation}
\end{lemma}

\begin{proof}
If \(R=+\infty\), there is nothing to prove. Assume \(R<+\infty\). Let
\[
    F_\varepsilon(\sigma)
    :=
    \{\omega\in\mc D_d:\|\omega-\sigma\|_1\le\varepsilon\} = \overline{N_\varepsilon(\sigma)}.
\]
Then \(F_\varepsilon(\sigma)\) is closed and $ N_\varepsilon(\sigma)\subseteq F_\varepsilon(\sigma)$. Hence
\[
    \widehat\mu_n^{\rho,K}(F_\varepsilon(\sigma))
    \ge
    \widehat\mu_n^{\rho,K}(N_\varepsilon(\sigma)),
\]
Since $\widehat \mu_n^{\rho,K}$ satisfies a LDP, we have 
\[
\inf_{\omega\in F_\varepsilon(\sigma)}
    I_K(\omega\|\rho)\le \liminf_{n\to\infty}
    -\frac1n\log
    \widehat\mu_n^{\rho,K}(F_\varepsilon(\sigma)) \le \limsup_{n\to \infty} -\frac{1}{n}\log  \widehat\mu_n^{\rho,K}(N_\varepsilon(\sigma)) \le R.
\]
We may choose \(\omega_m\in F_{\varepsilon_m}(\sigma)\) with $\varepsilon_m\to 0$, such that
\[
    I_K(\omega_m\|\rho)\le R+\frac1m.
\]
Since \(\omega_m\to\sigma\) and \(I_K(\cdot\|\rho)\) is lower semicontinuous,
\[
    I_K(\sigma\|\rho)
    \le
    \liminf_{m\to\infty}I_K(\omega_m\|\rho)
    \le R.
\]
This proves the lemma.
\end{proof}
Finally, we are able to prove Theorem~\ref{thm:main-covariant}:
\begin{proof}[Proof of Theorem~\ref{thm:main-covariant}]
    By Lemma~\ref{lem:local-bound-to-rate-bound}, we only need to show that for any $\varepsilon > 0$,
    \[
    \limsup_{n\to\infty}
    -\frac1n\log
    \widehat\mu_n^{\rho,K}(N_\varepsilon(\sigma)) \le S^{\mathrm{Keyl}}(\sigma \| \rho).
    \]
    Suppose $\sigma = U\diag(x) U^\dagger$, then 
    \[
    N_\varepsilon(\sigma) = \{\tau: \|\tau - \sigma\|_1 < \varepsilon\} = \{\tau: \|U^\dagger \tau U - \diag(x)\|_1 < \varepsilon\} = U N_\varepsilon(\diag(x)) U^\dagger.
    \]
    By definition of $\widehat \mu_n^{\rho,K}$, and left-invariance of Haar measure, we have
    \[
    \begin{aligned}
         & \widehat\mu_n^{\rho,K}(N_\varepsilon(\sigma)) \\
         & = \sum_{\lambda \vdash_d n} \dim \mc P_\lambda \dim \mc Q_\lambda\Tr(q_\lambda(\rho)\iint_{W\diag(y)W^\dagger \in N_\varepsilon(\sigma)} q_\lambda(W) K_{\lambda, y} q_\lambda(W^\dagger) d\nu_n(y) d\mu_{\mathrm{Haar}}(W)) \\
         & = \sum_{\lambda \vdash_d n} \dim \mc P_\lambda \dim \mc Q_\lambda\Tr(q_\lambda(\rho)\iint_{U^\dagger W\diag(y)W^\dagger U \in N_\varepsilon(\diag(x))} q_\lambda(W) K_{\lambda, y} q_\lambda(W^\dagger) d\nu_n(y) d\mu_{\mathrm{Haar}}(W)) \\
         & = \sum_{\lambda \vdash_d n} \dim \mc P_\lambda \dim \mc Q_\lambda\Tr(q_\lambda(\rho)\iint_{ W\diag(y)W^\dagger \in N_\varepsilon(\diag(x))}q_\lambda(U) q_\lambda(W) K_{\lambda, y} q_\lambda(W^\dagger) q_\lambda(U^\dagger) d\nu_n(y) d\mu_{\mathrm{Haar}}(W)) \\
         & = \sum_{\lambda \vdash_d n} \dim \mc P_\lambda \dim \mc Q_\lambda\Tr(q_\lambda(U^\dagger \rho U) K_{\lambda,x}^\varepsilon),
    \end{aligned}
    \]
    where $K_{\lambda,x}^\varepsilon$ is introduced in \eqref{eq:averaged-seed} which is torus invariant. Since $\widehat\mu_n^{\rho,K}$ is consistent, by Lemma~\ref{lem:consistent-implication}, there exists a Young sequence $\lambda^{(n)}$ with \[
    \lambda^{(n)}/n \to x,\quad \ell(\lambda^{(n)}) \le \rank(x),
    \]
    such that \[
    \lim_{n\to \infty} -\frac{1}{n} \log \Tr(q_{\lambda^{(n)}}(\diag(x)) K_{\lambda^{(n)},x}^\varepsilon) = H(x).
    \]
    Then by Lemma~\ref{lem:robust-Keyl-comparison-torus-invariant},
    \[
    \limsup_{n\to \infty} -\frac{1}{n} \log \Tr(q_{\lambda^{(n)}}(U^\dagger \rho U) K_{\lambda^{(n)},x}^\varepsilon)\le L_{\mathrm{Keyl}}(x,U^\dagger \rho U).
    \]
    Hence
    \[
    \begin{aligned}
        \widehat\mu_n^{\rho,K}(N_\varepsilon(\sigma)) & =\sum_{\lambda \vdash_d n} \dim \mc P_\lambda \dim \mc Q_\lambda\Tr(q_\lambda(U^\dagger \rho U) K_{\lambda,x}^\varepsilon) \\
        & \ge  \dim \mc P_{\lambda^{(n)}} \dim \mc Q_{\lambda^{(n)}} \Tr(q_{\lambda^{(n)}} (U^\dagger \rho U) K_{\lambda^{(n)},x}^\varepsilon).
    \end{aligned}
    \]
    Finally, using Lemma~\ref{lemma:rep-data}, $\dim \mc P_{\lambda^{(n)}} \approx \exp(n H(x))$ and $\dim \mc Q_{\lambda^{(n)}}$ is polynomial in $n$,  \[
    \limsup_{n\to\infty}
    -\frac1n\log
    \widehat\mu_n^{\rho,K}(N_\varepsilon(\sigma)) \le -H(x) + L_{\mathrm{Keyl}}(x,U^\dagger \rho U) = S^{\mathrm{Keyl}}(\sigma \| \rho).
    \]
\end{proof}

\appendix

\section{Properties of Keyl's rate function
} \label{section:Appendix_KeylsRateProperties}
This appendix records some basic properties of Keyl's rate function as described in \eqref{eq:Keyl-rate-intro}. First we characterize when $S^{\mathrm{Keyl}}(\sigma\|\rho)$ is infinite.  
\begin{proposition}\label{prop:Keyl-finiteness}
Let \(\sigma,\rho\in\mathcal D_d\), and write
\[
    \sigma=U\diag(x)U^\dagger,\quad x \in \Delta_d^\downarrow, \quad x_{d+1}:= 0.
\]
Then the following are equivalent:
\begin{enumerate}
    \item \label{condition:1} \(S^{\mathrm{Keyl}}(\sigma\|\rho)=+\infty\).
    \item \label{condition:2} There exists \(k\) with \(x_k>x_{k+1}\) such that \(\Delta_k(U^\dagger\rho U)=0\).
    \item \label{condition:3}\(\rank(\sigma)>\rank(\sigma\rho).\)
    \item \label{condition:4}$\supp(\sigma)\cap\ker(\rho)\neq\{0\}$.
\end{enumerate}
\end{proposition}
\begin{proof}
    From the formula in \eqref{eq:Keyl-rate-intro}, we directly see that $S^{\mathrm{Keyl}}(\sigma\|\rho) = +\infty$ if and only if Condition \ref{condition:2} holds. Note that Condition~\ref{condition:3} being equivalent to Condition~\ref{condition:4} follows from elementary linear algebra. Then we only need to show that Condition~\ref{condition:2} is equivalent to Condition~\ref{condition:4}. First, assume Condition~\ref{condition:2}: there exists $k$ with $x_k > x_{k+1}$ such that \(\Delta_k(U^\dagger\rho U)=0\), then there exists $z = (y,0,\cdots,0), y \in \mb C^k$ such that $z^* (U^\dagger \rho U)z = 0$ implying $Uz \in \ker(\rho)$. Because $x_k > x_{k+1} \ge 0$, and $z$ is supported on the first $k$ components, $Uz \in \supp(\sigma)$ implies $Uz \in \supp(\sigma)\cap\ker(\rho)$. Now assume Condition~\ref{condition:4}, choose $0\neq w \in \supp(\sigma)\cap\ker(\rho)$ and define $z:= U^\dagger w$, which has the form $z = (y,0,\cdots,0), y \in \mb C^r$ where $r = \rank(\sigma)$. Since $w \in \ker(\rho)$, we have $U^\dagger \rho U z = 0$ hence $z^*U^\dagger \rho Uz = 0$ implying $\Delta_r(U^\dagger\rho U)=0$, concluding Condition~\ref{condition:2}.
\end{proof}

Notice that the finiteness condition for $S^{\mathrm{Keyl}}(\sigma\|\rho)$ is weaker than the condition for $D(\sigma\|\rho)< +\infty$, in the sense that
\[
    \{(\sigma, \rho): S^{\rm{Keyl}}(\sigma\|\rho) = \infty\} \subsetneq \{(\sigma, \rho): D(\sigma\|\rho) = \infty\}.
\]
For example, taking $\sigma = \ket{0}\bra{0}$ and $\rho = \ket{+}\bra{+}$, we have $D(\sigma\|\rho) = \infty$, while $S^{\mathrm{Keyl}}(\sigma\|\rho) = - \log F(\sigma, \rho) < +\infty$. 

Now we are able to prove Proposition~\ref{prop:Keyl-vs-relative-entropy}:
\begin{proof}[Proof of Proposition~\ref{prop:Keyl-vs-relative-entropy}]
If $D(\sigma \| \rho) = +\infty$, then the equality condition directly follows from Proposition~\ref{prop:Keyl-finiteness}. Without loss of generality, we assume $D(\sigma \| \rho) < +\infty$.

Write $\sigma=\sum_{j=1}^d x_j |\psi_j\rangle\langle \psi_j|,\ x_1\ge \cdots\ge x_d$, and define
\[
    \Pi_k:=\sum_{j=1}^k |\psi_j\rangle\langle \psi_j|,
    \qquad
    x_{d+1}:=0.
\]
Then $\sigma=\sum_{k=1}^d (x_k-x_{k+1})\Pi_k$. Moreover, for $U = \sum_{j=1}^d |\psi_j\rangle \langle j |$, we have $\sigma = U \diag(x) U^\dagger$ and
\[
    \Delta_k(U^\dagger \rho U) =\Delta_k\left(\sum_{j,l=1}^d \langle \psi_j|\rho |\psi_l\rangle |j\rangle\langle l |\right) =  \det(\Pi_k\rho\Pi_k),
\]
where \(\Pi_k\rho\Pi_k\) is regarded as an operator fully supported on \(\Pi_k\mathcal H\) since $\supp(\sigma) \subseteq \supp(\rho)$. Hence
\[
    \log \Delta_k(U^\dagger\rho U)=\Tr\log(\Pi_k\rho\Pi_k).
\]
Therefore
\[
\begin{aligned}
    D(\sigma\|\rho)-S^{\mathrm{Keyl}}(\sigma\|\rho)
    &=
    \sum_{k=1}^d (x_k-x_{k+1})
    \Bigl[
        \Tr\log(\Pi_k\rho\Pi_k)
        -
        \Tr(\Pi_k\log\rho \Pi_k)
    \Bigr].
\end{aligned}
\]
Since \(\log\) is operator monotone, operator Jensen's inequality \cite{Hansen1979} gives
\[
    \log(\Pi_k\rho\Pi_k)\ge\Pi_k(\log\rho)\Pi_k.
\]
Taking traces, we have $D(\sigma\|\rho)\ge S^{\mathrm{Keyl}}(\sigma\|\rho)$.

Now consider equality condition. The equality condition of Jensen's inequality for the strictly concave function \(\log\) says that
\[
\log(\Pi_k\rho\Pi_k)=\Pi_k(\log\rho)\Pi_k
\]
holds if and only if $[\Pi_k,\rho]=0$. Since every coefficient \(x_k-x_{k+1}\) is nonnegative, equality
\[
    D(\sigma\|\rho)=S^{\mathrm{Keyl}}(\sigma\|\rho)
\]
holds if and only if
\[
    [\Pi_k,\rho]=0
    \qquad
    \text{for every } k \text{ with } x_k>x_{k+1}.
\]
Using $\sigma=\sum_{k=1}^d (x_k-x_{k+1})\Pi_k$, we have $[\sigma,\rho] = 0$, concluding the necessity part. For the sufficiency part, if \([\sigma,\rho]=0\), it is direct to see $S^{\mathrm{Keyl}}(\sigma\|\rho)= D(\sigma\|\rho)$.
\end{proof}

\section{General contraction principle of large deviations}\label{app:LDP}
In this section, we review a general contraction principle thus concluding the proof of LDP for the estimator defined in~\eqref{eq:estimator-relative-entropy}. 

Let \(\mc X,\mc Y\) be two Polish spaces, let \(\{\mu_n\}_{n\ge 1}\) be a sequence of
probability measures on \(\mc X\), and let $\Phi:\mc X\to \mc Y$ be a measurable map. For each $n\ge 1$, define $\nu_n:=\mu_n\circ \Phi^{-1}$ as a probability measure on \(\mc Y\). For any Borel set $A$, denote 
\[
\mathrm{int}(A),\quad \mathrm{cl}(A)
\]
as the interior set (maximal open set contained in $A$) and closure of A (minimal closed set containing $A$). 
The following result is a simpler version of Proposition 4.4 in \cite{Mariani2018}.
\medskip
\begin{theorem}\label{thm:contraction-LDP}
Suppose $\{\mu_n\}_{n\ge 1}$ satisfies a LDP with rate function $I(\cdot)$. Then for any open set \(O\subseteq \mc  Y\),
    \[
    \limsup_{n\to\infty}\left(-\frac1n \log \nu_n(O)\right)\le\inf_{y\in O} \overline J(y),
    \]
    where $N_\delta(y)$ is the open ball centered at $y$ with radius $\delta$ and
    \[
    \overline J(y):=\inf \left\{I(x):x\in\lim_{\delta\downarrow 0}\mathrm{int}\left(\Phi^{-1}\bigl(N_\delta(y)\bigr) \right)\right\}.
    \]
 If \(\{\mu_n\}\) is exponentially tight, i.e., $\inf_{\text{compact}\ K\subseteq \mc X} \limsup_{n\to \infty} \frac{1}{n} \log \mu_n(K^c) = -\infty$, then for any compact set \(K\subseteq \mc Y\),
    \[
    \liminf_{n\to\infty}\left(-\frac1n \log \nu_n(K)\right)\ge\inf_{y\in K} \underline J(y),
    \]
    where
    \[
    \underline J(y):=\inf\left\{I(x):\ x\in\lim_{\delta\downarrow 0}\mathrm{cl}\left(\Phi^{-1}\bigl(N_\delta(y)\bigr)\right)\right\}.
    \]
\end{theorem}
Using Theorem \ref{thm:contraction-LDP}, we can show that the estimator $\widehat \theta_n$ defined in~\eqref{eq:estimator-relative-entropy} satisfies a LDP. In fact, let $\mc X = \mc D_d \times [0,1]$, $\mc Y = \mc D_d$ and recall that $p_\tau = \Tr(\tau^{\otimes \ell} M_\ell)$,
\[
\Phi_{\omega_*}(\tau,p)=
    \begin{cases}
    \tau, & |p - p_\tau|\le\eta,\\
    \omega_*, & |p - p_\tau|>\eta,
     \end{cases}
\]
If $\sigma \neq \omega_*$, then for $\delta$ very small such that $\omega_* \not\in N_\delta(\sigma)$, we have 
\[
\Phi_{\omega_*}^{-1}(N_\delta(\sigma)) = \{(\tau,p): |p - p_\tau| < \eta,\ \tau \in N_\delta(\sigma)\}, 
\]
which implies
\[
\lim_{\delta\downarrow 0}\mathrm{int}\left(\Phi_{\omega_*}^{-1}\bigl(N_\delta(\sigma)\bigr)\right) = \{(\sigma,p): |p - p_\sigma|<\eta\},\quad \lim_{\delta\downarrow 0}\mathrm{cl}\left(\Phi_{\omega_*}^{-1}\bigl(N_\delta(\sigma)\bigr)\right) = \{(\sigma,p): |p - p_\sigma|\le \eta\}
\]
If $\sigma = \omega_*$, 
\[
\Phi_{\omega_*}^{-1}(N_\delta(\omega_*)) = \{(\tau,p): |p - p_\tau| < \eta,\ \tau \in N_\delta(\omega_*)\} \bigcup \{(\tau,p): |p - p_\tau| \ge \eta,\ \tau \in \mc D_d\}.
\]
This implies
\[
\begin{aligned}
    & \lim_{\delta\downarrow 0} \mathrm{int} \left(\Phi_{\omega_*}^{-1}(N_\delta(\omega_*))\right) = \{(\omega_*,p): |p - p_{\omega_*}|<\eta\} \bigcup \{(\tau,p): |p - p_\tau| > \eta\}, \\ 
    &\lim_{\delta\downarrow 0}\mathrm{cl} \left(\Phi_{\omega_*}^{-1}(N_\delta(\omega_*))\right) = \{(\omega_*,p): |p - p_{\omega_*}|\le \eta\} \bigcup \{(\tau,p): |p - p_\tau| \ge \eta\}
\end{aligned}
\]
Recall that $(\widehat \rho_n, \overline Z_n)$ satisfies a LDP with rate function $R(\sigma,p\|\rho)$ and for any fixed $\rho$, $R(\cdot,\cdot\|\rho)$ is continuous on $\mc D_d\times [0,1]$. Thus we get the same infimum over $\lim_{\delta\downarrow 0}\mathrm{int}\left(\Phi_{\omega_*}^{-1}\bigl(N_\delta(\sigma)\bigr)\right)$ and $\lim_{\delta\downarrow 0}\mathrm{cl}\left(\Phi_{\omega_*}^{-1}\bigl(N_\delta(\sigma)\bigr)\right)$, implying $\overline J = \underline J$ thus a LDP for $\widehat \theta_n$.

\bibliography{optimal}

\providecommand{\etalchar}[1]{$^{#1}$}
\begin{thebibliography}{HCMT{\etalchar{+}}26}

\bibitem[AD15]{audenaert2015alpha}
K.~M.~R. Audenaert and N.~Datta.
\newblock {\em ``{$\alpha$-$z$-R{\'e}nyi relative entropies}''}.
\newblock \href{http://dx.doi.org/10.1063/1.4906367}{Journal of Mathematical Physics {\bf 56}(2):\,022202} (2015).

\bibitem[BBH21]{berta2021composite}
M.~Berta, F.~G. S.~L. Brand{\~a}o, and C.~Hirche.
\newblock {\em ``{On composite quantum hypothesis testing}''}.
\newblock \href{http://dx.doi.org/10.1007/s00220-021-04133-8}{Communications in Mathematical Physics {\bf 385}(1):\,55--77} (2021).

\bibitem[BCV21]{botero2021large}
A.~Botero, M.~Christandl, and P.~Vrana.
\newblock {\em ``{Large deviation principle for moment map estimation}''}.
\newblock \href{http://dx.doi.org/10.1214/21-EJP636}{Electronic Journal of Probability {\bf 26}:\,1--23} (2021).

\bibitem[Chr06]{Christandl2006BipartiteStates}
M.~Christandl.
\newblock {\em ``{The Structure of Bipartite Quantum States -- Insights from Group Theory and Cryptography}''}, (2006).
\newblock Available online: \url{https://arxiv.org/abs/quant-ph/0604183}.

\bibitem[CM06]{ChristandlMitchison2006}
M.~Christandl and G.~Mitchison.
\newblock {\em ``The spectra of density operators and the Kronecker coefficients of the symmetric group''}.
\newblock \href{http://dx.doi.org/10.1007/s00220-005-1435-1}{Communications in Mathematical Physics {\bf 261}(3):\,789--797} (2006).

\bibitem[FH91]{FultonHarris1991}
W.~Fulton and J.~Harris.
\newblock {\em Representation Theory: A First Course}.
\newblock volume 129 of {\em Graduate Texts in Mathematics}, \href{http://dx.doi.org/10.1007/978-1-4612-0979-9}{Springer} (1991).

\bibitem[FMVW25]{frenkel2025errorbounds}
P.~E. Frenkel, M.~Mosonyi, P.~Vrana, and M.~Weiner.
\newblock {\em ``Error bounds for composite quantum hypothesis testing and a new characterization of the weighted Kubo-Ando geometric means''}, (2025).
\newblock Available online: \url{https://arxiv.org/abs/2503.13379}.

\bibitem[GW09]{GoodmanWallach2009}
R.~Goodman and N.~R. Wallach.
\newblock {\em Symmetry, Representations, and Invariants}.
\newblock volume 255 of {\em Graduate Texts in Mathematics}, \href{http://dx.doi.org/10.1007/978-0-387-79852-3}{Springer} (2009).

\bibitem[Han79]{Hansen1979}
F.~Hansen.
\newblock {\em ``An Operator Inequality.''}.
\newblock \href{http://dx.doi.org/10.1007/BF01371046}{Mathematische Annalen {\bf 246}:\,249--250} (1979).

\bibitem[Har05]{Harrow2005SchurTransform}
A.~W. Harrow.
\newblock {\em ``{Applications of coherent classical communication and the Schur transform to quantum information theory}''}, (2005).
\newblock Available online: \url{https://arxiv.org/abs/quant-ph/0512255}.

\bibitem[Hay01]{Hayashi_2001}
M.~Hayashi.
\newblock {\em ``Asymptotics of quantum relative entropy from a representation theoretical viewpoint''}.
\newblock \href{http://dx.doi.org/10.1088/0305-4470/34/16/309}{Journal of Physics A: Mathematical and General {\bf 34}(16):\,3413–3419} (2001).

\bibitem[Hay17]{hayashi2017}
M.~Hayashi.
\newblock {\em A group theoretic approach to quantum information}.
\newblock \href{http://dx.doi.org/10.1007/978-3-319-45241-8}{Springer} (2017).

\bibitem[Hay25]{hayashi2025another}
M.~Hayashi.
\newblock {\em ``Another quantum version of sanov theorem''}.
\newblock \href{http://dx.doi.org/10.1007/s00023-025-01612-9}{Annales Henri Poincar{\'e} pages 1--22} (2025).

\bibitem[HCMT{\etalchar{+}}26]{hu2026sample}
Y.~Hu, E.~Cervero-Mart{\'\i}n, E.~Theil, L.~Man{\v{c}}inska, and M.~Tomamichel.
\newblock {\em ``Sample-optimal and memory-efficient quantum state tomography''}.
\newblock \href{http://dx.doi.org/https://doi.org/10.1103/z93d-tw3m}{Physical Review A {\bf 113}(5):\,052446} (2026).

\bibitem[HF26]{hayashi2026operational}
M.~Hayashi and K.~Fang.
\newblock {\em ``{Operational interpretation of the reverse sandwiched R{\'e}nyi divergences in composite quantum hypothesis testing}''}, (2026).
\newblock Available online: \url{https://arxiv.org/abs/2605.02203}.

\bibitem[HHJ{\etalchar{+}}17]{Haah_2017}
J.~Haah, A.~W. Harrow, Z.~Ji, X.~Wu, and N.~Yu.
\newblock {\em ``Sample-optimal tomography of quantum states''}.
\newblock \href{http://dx.doi.org/10.1109/tit.2017.2719044}{IEEE Transactions on Information Theory {\bf 63}(9):\,5628 -- 5641} (2017).

\bibitem[HM02]{HayashiMatsumoto2002}
M.~Hayashi and K.~Matsumoto.
\newblock {\em ``Quantum universal variable-length source coding''}.
\newblock \href{http://dx.doi.org/10.1103/PhysRevA.66.022311}{Physical Review A {\bf 66}:\,022311} (2002).

\bibitem[Hol00]{hollander2000large}
F.~Hollander.
\newblock {\em Large deviations}.
\newblock volume~14, \href{http://dx.doi.org/https://doi.org/10.1090/fim/014}{American Mathematical Soc.} (2000).

\bibitem[HP91]{hiai1991proper}
F.~Hiai and D.~Petz.
\newblock {\em ``{The proper formula for relative entropy and its asymptotics in quantum probability}''}.
\newblock \href{http://dx.doi.org/10.1007/BF02100287}{Communications in Mathematical Physics {\bf 143}(1):\,99--114} (1991).

\bibitem[JR25]{ji2025beyond}
K.~Ji and B.~Regula.
\newblock {\em ``{Beyond Hoeffding and Chernoff: Trading conclusiveness for advantages in quantum hypothesis testing}''}, (2025).
\newblock Available online: \url{https://arxiv.org/abs/2510.07601}.

\bibitem[Key06]{Keyl_2006}
M.~Keyl.
\newblock {\em ``Quantum state estimation and large deviations''}.
\newblock \href{http://dx.doi.org/10.1142/s0129055x06002565}{Reviews in Mathematical Physics {\bf 18}(01):\,19–60} (2006).

\bibitem[KW01]{KeylWerner2001}
M.~Keyl and R.~F. Werner.
\newblock {\em ``Estimating the spectrum of a density operator''}.
\newblock \href{http://dx.doi.org/10.1103/PhysRevA.64.052311}{Physical Review A {\bf 64}(5):\,052311} (2001).

\bibitem[Lam25]{lami2025generalised}
L.~Lami.
\newblock {\em ``Generalised quantum Sanov theorem revisited''}, (2025).
\newblock Available online: \url{https://arxiv.org/abs/2510.06340}.

\bibitem[LBCR{\etalchar{+}}24]{lipka2024quantum}
P.~Lipka-Bartosik, C.~T. Chubb, J.~M. Renes, M.~Tomamichel, and K.~Korzekwa.
\newblock {\em ``{Quantum dichotomies and coherent thermodynamics beyond first-order asymptotics}''}.
\newblock \href{http://dx.doi.org/10.1103/PRXQuantum.5.020335}{PRX Quantum {\bf 5}(2):\,020335} (2024).

\bibitem[Led25]{Leditzky2025RepTheoryQI}
F.~Leditzky.
\newblock {\em ``{Representation-theoretic methods in quantum information theory}''}.
\newblock Lecture notes, (2025).
\newblock Available online: \url{https://felixleditzky.info/files/Felix%20Leditzky%20-%20Math%20595%20Representation-theoretic%20methods%20in%20quantum%20information%20theory.pdf}.

\bibitem[Mar18]{Mariani2018}
M.~Mariani.
\newblock {\em ``A {$\Gamma$}-convergence approach to large deviations''}.
\newblock \href{http://dx.doi.org/https://doi.org/10.2422/2036-2145.201301_010}{Annali della Scuola Normale Superiore di Pisa, Classe di Scienze {\bf 18}(3):\,951--976} (2018).

\bibitem[MSW22]{mosonyi2020error}
M.~Mosonyi, Z.~Szil{\'a}gyi, and M.~Weiner.
\newblock {\em ``{On the error exponents of binary state discrimination with composite hypotheses}''}.
\newblock \href{http://dx.doi.org/10.1109/TIT.2021.3125683}{IEEE Transactions on Information Theory {\bf 68}(2):\,1032--1067} (2022).

\bibitem[N{\"o}t15]{notzel2015class}
J.~N{\"o}tzel.
\newblock {\em ``{A class of permutation-invariant measurements and their relation to quantum relative entropies}''}, (2015).
\newblock Available online: \url{https://arxiv.org/abs/1507.04609}.

\bibitem[ON00]{ogawa2000strong}
T.~Ogawa and H.~Nagaoka.
\newblock {\em ``Strong converse and Stein's lemma in quantum hypothesis testing''}.
\newblock \href{http://dx.doi.org/10.1109/18.887855}{IEEE Transactions on Information Theory {\bf 46}(7):\,2428--2433} (2000).

\bibitem[OW16]{o2016efficient}
R.~O'Donnell and J.~Wright.
\newblock {\em ``Efficient quantum tomography''}.
\newblock \href{http://dx.doi.org/https://doi.org/10.1145/2897518.2897544}{Proceedings of the forty-eighth annual ACM symposium on Theory of Computing pages 899--912} (2016).

\bibitem[OW21]{ODonnell2021}
R.~O'Donnell and J.~Wright.
\newblock {\em ``{Learning and Testing Quantum States via Probabilistic Combinatorics and Representation Theory}''}.
\newblock Survey note, (2021).
\newblock Available online: \url{https://www.cs.cmu.edu/~odonnell/papers/learning-quantum-states.pdf}.

\bibitem[PSTW25]{pelecanos2025mixed}
A.~Pelecanos, J.~Spilecki, E.~Tang, and J.~Wright.
\newblock {\em ``Mixed state tomography reduces to pure state tomography''}, (2025).
\newblock Available online: \url{https://arxiv.org/abs/2511.15806}.

\bibitem[PSW25]{pelecanos2025debiased}
A.~Pelecanos, J.~Spilecki, and J.~Wright.
\newblock {\em ``The debiased Keyl's algorithm: a new unbiased estimator for full state tomography''}, (2025).
\newblock Available online: \url{https://arxiv.org/abs/2510.07788}.

\bibitem[Sag01]{Sagan2001}
B.~E. Sagan.
\newblock {\em The Symmetric Group: Representations, Combinatorial Algorithms, and Symmetric Functions}.
\newblock volume 203 of {\em Graduate Texts in Mathematics}, \href{http://dx.doi.org/10.1007/978-1-4757-6804-6}{Springer} (2001).

\bibitem[Sch01]{schur1901klasse}
I.~Schur.
\newblock {\em {\"U}ber eine Klasse von Matrizen, die sich einer gegebenen Matrix zuordnen lassen}.
\newblock  (1901).

\bibitem[SSW25]{scharnhorst2025optimal}
T.~Scharnhorst, J.~Spilecki, and J.~Wright.
\newblock {\em ``Optimal lower bounds for quantum state tomography''}, (2025).
\newblock Available online: \url{https://arxiv.org/abs/2510.07699}.

\bibitem[Wey46]{weyl1946classical}
H.~Weyl.
\newblock {\em The classical groups: their invariants and representations}.
\newblock volume~1, \href{http://dx.doi.org/https://doi.org/10.2307/j.ctv3hh48t}{Princeton university press} (1946).

\end{thebibliography}
\end{document}